\documentclass[a4,11pt]{article}

\usepackage[a4paper,top=2cm,bottom=2cm,left=2cm,right=2cm]{geometry}
\usepackage[toc]{appendix}
\usepackage{latexsym} 	
\usepackage{amsmath}  
\usepackage{amsthm}	
\usepackage{amsopn}	
\usepackage{amsthm} 
\newtheorem{definition}{Definition}

\usepackage{amscd}	
\usepackage{amssymb} 
\usepackage{array}	
\usepackage{actuarialsymbol} 
\usepackage{mathtools} 
\usepackage[utf8x]{inputenc} 
\usepackage{makeidx} 
\usepackage{multicol} 
\usepackage{multirow} 
\usepackage{caption} 
\usepackage{float} 
\usepackage{enumerate} 
\usepackage{graphicx} 
\usepackage{subfig} 
\usepackage{booktabs} 
\usepackage{soul}  
\usepackage{color}  
\usepackage[italian,english]{babel}  
\usepackage{wasysym}
\usepackage{appendix} 
\usepackage{longtable}
\usepackage{microtype}  \graphicspath{ {Figures/}} 
\usepackage{xcolor}
\usepackage{enumitem}
\usepackage{graphicx}
\usepackage{wrapfig}
\usepackage{lscape}
\usepackage{rotating}
\usepackage{epstopdf}
\usepackage{bm}
\usepackage{adjustbox}
\usepackage{afterpage}
\usepackage[authoryear]{natbib}
\usepackage{hyperref}
\usepackage{soul,xcolor}
\theoremstyle{proposition}
\newtheorem*{proposition}{Proposition}

\title{\textbf{Deepening Lee-Carter for longevity projections with uncertainty estimation}}
\author{Mario Marino\\\normalsize Department of Statistics, Sapienza University of Rome \and Susanna Levantesi\\\normalsize Department of Statistics, Sapienza University of Rome \and Andrea Nigri\\\normalsize Department of Statistics, Sapienza University of Rome}

\date{\today}

\begin{document}

	\maketitle

\begin{abstract}
Undoubtedly, several countries worldwide endure to experience a continuous increase in life expectancy, extending the challenges of life actuaries and demographers in forecasting mortality. Although several stochastic mortality models have been proposed in past literature, the mortality forecasting research remains a crucial task. Recently, various research works encourage the adequacy of deep learning models to extrapolate suitable pattern within mortality data. Such a learning models allow to achieve accurate point predictions, albeit also uncertainty measures are necessary to support both model estimates reliability and risk evaluations. To the best of our knowledge, machine and deep learning literature in mortality forecasting lack for studies about uncertainty estimation.
As new advance in mortality forecasting, we formalizes the deep Neural Networks integration within the Lee-Carter framework, posing a first bridge between the deep learning and the mortality density forecasts. We test our model proposal in a numerical application considering three representative countries worldwide and both genders, scrutinizing two different fitting periods. Exploiting the meaning of both biological reasonableness and plausibility of forecasts, as well as performance metrics, our findings confirm the suitability of deep learning models to improve the predictive capacity of the Lee-Carter model, providing more reliable mortality boundaries also on the long-run. 
\\ \vspace{0.2 cm}
\end{abstract}
\begin{flushleft}
	\vspace{0.3 cm}
\end{flushleft}

\section{Introduction}
\label{sec:introduction}
Since the second half of the 20th century, mortality has exhibited notable improvements engaging attention from life insurers and pension systems, as well as from actuarial and demographic researchers. Principally, mortality reductions in modern populations arise by virtue of a continuous flow of social progresses (\cite{oeppen}). In fact, industrialized countries demonstrated progress in socio-economic dynamics, health systems and lifestyles, impacting on how mortality will vary in the future. 
Various factors move human longevity trends and different mortality scenarios should be anticipated through predictive analysis. The need of accurate forecasting to address longevity risk and adequately pricing the annuities products has led actuaries towards more sophisticated extrapolative methods, in a stochastic environment, see for instance \cite{LeeCarter92}, \cite{Brouhns}, \cite{RenshawHaberman06}, \cite{Cairns2006}, \cite{booth2008}, \cite{Cairns2009}, \cite{plat}, \cite{hunt2014} and \cite{currie2016}.

Demographers and actuaries have concentrated their efforts on the model functional form and its parametrization in order to better explain the mortality structure. In most of these models, mortality projections arise from time-dependent parameters, modelled by time series analysis techniques, the class of ARIMA processes among all. However, alternative mortality forecasting methods have been suggested in past literature. For instance, a P-spline based approach is proposed in \cite{currie2004}, where forthcoming values are interpreted as missing-value findable by smoothing procedures. A development of this model is presented in \cite{camarda}, overcoming robustness forecasting problems. An innovative proposal has been introduced in \cite{mitchell}, wherein the Lee-Carter (henceforth LC) time-index is predicted through a Normal Inverse Gaussian distribution, attaining accuracy in the approximation of the observed force of mortality. Furthermore, new advances in mortality modelling, grounded in the machine and deep learning models, has recently appeared in literature. The first insight based on machine learning tools is offered in \cite{deprez}, where regression trees algorithms are adopted to improve the estimation of death rates from canonical models, such as the LC and the Renshaw-Haberman one.  
These findings are extended in \cite{levpizzo} and \cite{levnigri} for predictive purposes. 
A Neural Network (henceforth NN) design for mortality analysis is initially scrutinized by \cite{hainaut}, profitably aiming to extrapolate suitable non-linearities in the observed force of mortality. A NN vision within the LC framework is presented in \cite{NLMSP2019}, \cite{perla} and \cite{richwut}. The former employs a Recurrent NN architecture, namely Long Short-Term Memory (henceforth LSTM), to model the future LC time-dependent parameter values. For each country investigated and both genders, numerical experiments performed confirm greater LSTM accuracy w.r.t. the best ARIMA process. The latter proposes a NN representation for the multi-population LC model, overcoming parameters optimization problems and achieving reliable forecasting performances. Following this wake, \cite{perla} takes the moves showing the remarkable accuracy achieved in a large-scale prediction of mortality. In particular, different NN structures are tested, such as the LSTM and the convolutional NN, engaging each of them to produce point forecasts of mortality rates simultaneously for many countries. Given the recurrent networks suitability for forecasting purposes, also \cite{NLM2020} consider a LSTM model to predict both life expectancy and lifespan disparity measure, concerning various countries and both genders. 

Deep learning models, especially recurrent NNs, are gaining confidence in many forecasting tasks, as well as in mortality. They are dynamic systems stemming from composition and superposition of non-linear functions, earning notable accuracy gains in predictive issues. Wanting to exploit the latter feature, we aim to investigate the suitability of deep NNs models within the LC framework to extrapolate the future mortality realizations.  
Contextualizing suggestions expressed in \cite{M4}, our approach pursues a models integrating deep learning techniques in the spirit of \cite{NLMSP2019}, representing an appropriate compromise between interpretation of the mortality model and high accuracy in projections. Therefore, we freeze the LC age-period mortality representation, forecasting the mortality profile employing a deep NN models.

It is worth to recall that a proper forecasting model provides robust point predictions, outlining the future  mortality trend, as well as confidence ranges of variability. Uncertainty measures associated with the expected values are necessary to sufficiently inspect the phenomenon and, at the same time, to judge both the model adequacy and the results reliability. As in actuarial assessments, uncertainty measures, such as prediction intervals, are imperative. This is a compelling topic, since learning models such as NNs furnish only point predictions. 
To this purposes, \cite{khosravi} provided an extensive methodological review of the main approaches for calculating confidence and prediction intervals, concluding that no method beats the other ones in each considered comparison metric. Anyhow, procedures based on structural  assumptions, such as the Delta method (\cite{wild}), the Mean-Variance Estimation (\cite{nix}) and the Bayesian approach (\cite{mackay}), are relevant solutions but suffering computational troubles that could be prohibitive. At the state of the art, the prevailing approach to forecast prediction intervals for NNs is based on coherent sampling techniques, favouring the estimation of a theoretical probability distribution through an empirical one, see for instance \cite{Tibsh}, \cite{Hesk}, \cite{khosravi}, \cite{Mazloumi}, \cite{Kasiviswanathan}, \cite{khosravi2015} and \cite{lietal}. In particular, bootstrap procedures seem to represent the more tempting alternative since they do not require stringent sampling assumptions, allowing for accurate plug-in estimates (\cite{efron}). In fact, such an approach have become a common practice to measure uncertainty in stochastic mortality models, as emerged in \cite{Brouhns2005}, \cite{Koissi}, \cite{li2009}, \cite{damato2011}, \cite{damato2012a} and \cite{damato2012b}. However, to the best of our knowledge, machine and deep learning literature in mortality forecasting lack for studies about uncertainty estimation.

As new advance in mortality forecasting, the present work formalizes the integration of deep learning techniques in the LC model framework, in terms of both point estimates and prediction intervals for future mortality rates. In doing so, we refer to a recurrent NN with LSTM architecture to forecast the LC time-index. The resulting integrated model, namely LC-LSTM, fill the gap between the deep learning integrated mortality models and the uncertainty estimation, getting suitable ranges of variability.We test the proposed model in a numerical application considering three countries worldwide, Australia, Japan and Spain, for both genders scrutinizing two different learning periods to deepen how they could affect the forecasting performances. The obtained results are assessed considering both qualitative and quantitative criteria. The former are well-established in \cite{Cairns2011}, concerning: (a) the biological reasonableness of mortality forecasts; (b) the plausibility of projected uncertainty at different ages; (c) the predictions robustness w.r.t. the historical mortality trend. Quantitative criteria, as performance metrics, are considered to judge achieved mortality forecasts with a backtesting approach. Our findings confirm the LC-LSTM suitability to produce plausible mortality projections, improving the LC predictive capacity, in particular in the long-run. The proposed framework might represent a prominent practice in the field of longevity forecasting, as for actuarial business tasks. 

The reminder of the paper is structured as follow. Section \ref{sec:NN} presents the RNN model with LSTM architecture. Section \ref{sec:LC_LSTM} illustrates the LC-LSTM model formalization. Section \ref{sec:pred_interval} discusses the uncertainty framework within the LC-LSTM, highlighting the methodology to estimate prediction intervals. Section \ref{sec:metriche} describes the suitable performance metrics to evaluate both point and interval forecasts. Section \ref{sec:numerical_application} reports the results, and related comments, from the LC-LSTM application to mortality data. Finally, Section \ref{sec:conclusioni} poses the conclusions.   

\bigskip

\section{The Neural Network model}
\label{sec:NN}

A NN model is a computational graph consisting of connected nodes, the neurons, located in consecutive layers. 
Connections among neurons are pondered by parameters, whose values are learned from data implementing efficient optimization procedures (\cite{rum}). Each neuron receives weighted information, namely activation, and transforms it employing a differentiable function, the activation function. As a consequence, NN outputs descend from composition and superposition of differentiable functions, providing flexible data-driven tools that deeply gather data features and generalize them. 

For forecasting purposes, specific NNs models, namely Recurrent Neural Networks (henceforth RNNs), are used to handle sequential data such as time series. 
In RNNs recurrent connections between neurons are added, so that the network processes data creating a dynamic memory. However, RNNs learning optimization could be tricky because of the vanishing or exploding gradient problems (\cite{pascanu}). 
To address such a problem, \cite{hocreiter} proposed the LSTM architecture, whose more engineered structure relies both on a memory block and gates, essentials to control data elaborations. In the following, we will consider the RNN with LSTM architecture, referring the interested reader to \cite{bengio}, \cite{aggarwal} and references therein for further details on RNNs and LSTM. 

\subsection{RNNs with LSTM architecture}
\label{sec:lstm_model}
In order to define the general structure of the RNN with LSTM architecture, let $N_0$ be the number of neurons within the input layer, $N_p$ the number of neurons of the $p^{th}$ hidden layer with $p \in \left\lbrace1,\ldots,P\right\rbrace$, and $N_{P+1}$ the number of neurons of the output layer. We have $P, N_0, N_p, N_{P+1} \in \mathbb{N}$.
Let $A^{(p)}: \mathbb{R}^{N_{p-1}} \to \mathbb{R}^{N_{p}}$ be an affine map defining the $p^{th}$ hidden layer activation, given the output produced by the $(p-1)^{th}$ hidden layer, and let $\phi : \mathbb{R}^{N_{p}} \to \mathbb{R}^{N_{p}}$ be a differentiable activation function.
%
\begin{definition}
	The output of a LSTM neuron at time $t$ in the $p^{th}$ hidden layer is:
	\begin{equation}
	\quad  H_t^{(p)} = \bm{o}_t^{(p)} \odot \text{tanh}\left(\bm{c}_t^{(p)}\right),
	\end{equation}
	where the components of the element-wise product stem from the LSTM neuron internal forward flow described by the following equations:
	\begin{equation}
	\label{eq:LSTM}
	\begin{aligned}
	\text{Forget gate}:&  \quad   \bm{f}_t^{(p)} = \sigma_f \circ A^{(p)} = \sigma \left(\langle\bm{W}_f^{(p)},H_t^{(p-1)}\rangle+\langle\bm{U}_f^{(p)},H_{t-1}^{(p)}\rangle+\bm{b}_f^{(p)}\right), \\
	\text{Input gate}:&  \quad  \bm{i}_t^{(p)} = \sigma_i \circ A^{(p)} = \sigma\left(\langle\bm{W}_i^{(p)},H_t^{(p-1)}\rangle+\langle\bm{U}_i^{(p)},H_{t-1}^{(p)}\rangle+\bm{b}_i^{(p)}\right),\\
	\text{Output gate}:&  \quad  \bm{o}_t^{(p)} = \sigma_o \circ A^{(p)} = \sigma\left(\langle\bm{W}_o^{(p)},H_t^{(p-1)}\rangle+\langle\bm{U}_o^{(p)},H_{t-1}^{(p)}\rangle+\bm{b}_o^{(p)}\right), \\
	\text{Memory state}:&  \quad  \bm{c}_t^{(p)} = \bm{f}_t^{(p)} \odot \bm{c}_{t-1}^{(p)} + \bm{i}_t^{(p)} \odot \text{tanh} \left(\langle\bm{W}_c^{(p)},H_t^{(p-1)}\rangle+\langle\bm{U}_c^{(p)},H_{t-1}^{(p)}\rangle+\bm{b}_c^{(p)}\right),\\
	\end{aligned}
	\end{equation}
	where $\sigma(x)=\left(1+e^x\right)^{-1}$ is the sigmoid function, $tanh(x)=\left(e^x-e^{-x}\right)\left(e^x+e^{-x}\right)^{-1}$ is the hyperbolic tangent function, $\left(\bm{W}_l^{(p)}, l=f,i,o,c \right)$ are the weight matrices for gates feedforward connections,  $\left(\bm{U}_l^{(p)}, l=f,i,o,c\right)$ are the weight matrices for gates recurrent connections and $\left( \bm{b}_l^{(p)}, l=f,i,o,c \right)$ are the bias terms. \vspace{0.2 cm}
\end{definition}
\begin{definition}
	Let $\mathcal{D}=\left\lbrace\left(\bm{x}_t,\bm{y}_t\right), \bm{x}_t \in \mathbb{R}^{N_0}, \bm{y}_t \in \mathbb{R}^{N_{P+1}} \right\rbrace$ be a dataset wherein $\bm{x}_t$ is the input variable at time $t$ and $\bm{y}_t$ the associated response. A RNN with LSTM architecture is a function $f_{LSTM}:\mathbb{R}^{N_0}\to\mathbb{R}^{N_{P+1}}$ such that:
	\begin{equation}
    \label{eq:def_LSTM}
    \bm{y}_t=f_{LSTM}\left(\bm{x}_t;\bm{\mathcal{W}}\right) + \gamma_t=\psi \circ\left(H_t^{(P)}\circ H_t^{(P-1)}\circ \cdots \circ H_t^{(1)}\right)\left(\bm{x}_t;\bm{\mathcal{W}}\right)+ \gamma_t,
	\end{equation}
	where $\psi:\mathbb{R}^{N_P}\to\mathbb{R}^{N_{P+1}}$ is the output layer activation function, $\bm{\mathcal{W}}$ is the set of all NN parameters and $\gamma_t$ is a noise term, having zero mean and variance $\sigma^2_{\gamma}$.
\end{definition}

Starting from eq.(\ref{eq:def_LSTM}), we pose a bridge between deep learning and mortality forecasting. Favouring the LC model as the bridge's structure, the following Sections formalize the resulting integrated mortality model, both in terms of point and interval estimates.

\bigskip

\section{The LC-LSTM model}
\label{sec:LC_LSTM}
Let us consider the LC Poisson model proposed in \cite{Brouhns} as the reference model describing the behavior of the age-period mortality rates. Hence, for ages $x \in \mathcal{X}=\left\lbrace0,1,\ldots,\omega \right\rbrace$ and calendar years $t \in \mathcal{T}=\left\lbrace t_0, t_1,\ldots, t_n \right\rbrace$, we assume that the observed number of deaths, $D_{x,t}$, follows a Poisson distribution:
\begin{equation}
\label{eq:hp_decessi}
D_{x,t} \sim Poi(E^c_{x,t}m_{x,t}), \\
\end{equation}
where $E^c_{x,t}$ is the central exposure to the death risk and $m_{x,t}=\mathbb{E}\left(\frac{D_{x,t}}{E^c_{x,t}}\right)$ is the central death rate. The equation defining the LC model structure associated to the assumption (\ref{eq:hp_decessi}) (\cite{currie2016}) is:
\begin{equation}
\label{eq:predictor_LC}
\log m_{x,t}=\alpha_x+\beta_x k_t, \\
\end{equation}
where $\alpha_x$ and $\beta_x$ are age-dependent parameters illustrating the mortality age pattern and $k_t$ is a time-index parameter representing the mortality behaviour over time. As is well-known, parameters constraints must be satisfied to ensure model identification, i.e. $\sum_{t=t_0}^{t_n} k_t =0$ and $\sum_{x=0}^{\omega}b_x=1$.  

Let $\bm{\kappa}_\mathcal{T}=(k_{t-j})_{t\in \mathcal{T}}$ be the vector of the time lagged $k_t$, being $j\in \mathbb{N}$ the time lag.
According to eq.(\ref{eq:def_LSTM}), we model the LC time-index as below:
\begin{equation}
\label{eq:parametric_eq_lstm}
k_t=f_{LSTM}\left(\bm{\kappa}_\mathcal{T}; \bm{\mathcal{W}}\right)+ \gamma_t=\psi \circ\left(H_t^{(P)}\circ H_t^{(P-1)}\circ \cdots \circ H_t^{(1)}\right)\left(\bm{\kappa}_\mathcal{T}; \bm{\mathcal{W}}\right)+ \gamma_t.
\end{equation}   
Integrating eq.(\ref{eq:parametric_eq_lstm}) within the LC structure in eq.(\ref{eq:predictor_LC}), the LSTM will act as a predictor over the forecasting horizon $\mathcal{T'}=\left\lbrace t_{n}+1, t_{n}+2,\ldots, t_{n}+s \right\rbrace$, and the LC-LSTM model expression is: 
\begin{equation}
\label{eq:linear_predictor_LC_LSTM}
\log m_{x,t}=\hat{\alpha}_x+\hat{\beta}_x\left(f_{LSTM}\left(\bm{\kappa}_\mathcal{T'}; \bm{\mathcal{W}}\right)+ \gamma_t\right), \quad \forall t\in \mathcal{T'},
\end{equation}
with $\hat{\alpha}_x$ and $\hat{\beta}_x$ the estimates of age-dependent parameters. 

The meaning of the proposed model integration is the following. As the mortality dynamic over time stems from a continuous evolution of various social and demographic factors, a coherent mortality profile investigation suggests an autoregressive approach to the time-index modeling. From a general perspective, the LC time-index values should be interpreted as the realization of the following process:
\begin{equation}
\label{eq:kt_generica}
k_t=\varphi\left(\bm{\kappa}_\mathcal{T}\right)+\gamma_t, \quad \forall t\in \mathcal{T},
\end{equation}
where the unknow function $\varphi:\mathbb{R}^{(t-j)} \to \mathbb{R}$  maps the vector $\bm{\kappa}_\mathcal{T}$ to $k_t$ over the time horizon $\mathcal{T}$, unless the noise component.
Referring to the RNNs universal functional approximation property (\cite{schafer}), the proposed model integration allows to resemble the unknown map $\varphi\left(\bm{\kappa}_\mathcal{T}\right)$ through a RNN with LSTM architecture, whose functional form is shaped according to the available time-index history. 
As the RNN model approximates the map $\varphi\left(\bm{\kappa}_\mathcal{T}\right)$, it also defines the mean of response variable conditioned to the explicative ones (\cite{bishop}), that is: 
\begin{equation}
\label{eq:stima_kt_lstm}
\hat{k}_t=\hat{f}_{LSTM}\left(\bm{\kappa}_\mathcal{T};\hat{\bm{\mathcal{W}}}\right) = \mathbb{E}\left(k_t | \bm{\kappa}_\mathcal{T}\right),
\end{equation}
where $\hat{f}_{LSTM}$ is the fitted function composition and $\hat{\bm{\mathcal{W}}}$ is the NN parameters estimate. Such a relation highlights that the LSTM model captures the LC time-index conditional expectation. Therefore, the LC-LSTM model provides the following point predictions:
\begin{equation}
\label{eq:point_prediction_LC_LSTM}
\log \hat{m}_{x,t}=\mathbb{E}\left(\log m_{x,t}\right)=\hat{\alpha}_x+\hat{\beta}_x\hat{f}_{LSTM}\left(\bm{\kappa}_\mathcal{T'}; \hat{\bm{\mathcal{W}}}\right), \quad \forall t\in \mathcal{T'}.
\end{equation}   
However, point predictions do not describe the uncertainty arising from the estimates of mortality rates. 
Therefore, a methodology for building prediction intervals are necessary in order to provide a measure of prediction uncertainty.

\section{Prediction intervals for the LC-LSTM model}
\label{sec:pred_interval}
Prediction intervals (henceforth PI) outline a probabilistic range suitable to incorporate various forecasting scenarios, then probing uncertainty on the future mortality realizations. Stochastic mortality models forecast PIs, whose estimates act as uncertainty measure linked to the expected future mortality(see for instance \cite{booth2008}, \cite{Cairns2009}, \cite{Cairns2011}, \cite{dowd}). Thus, in a proper forecasting process PIs are meaningful in supporting both risk evaluations and the model estimates reliability. Referring to NNs, PIs construction is a challenging task because of different uncertainty sources impact on the learning process, then conditioning the NN generalization performances.   
By a broad perspective, NNs models are exposed to a learning uncertainty, depending both on the data and the NN functioning.
Since the data employed in the learning process are a realization of an underlying stochastic process, a training data uncertainty looms. Indeed, varying input could involve in distinct function compositions. In addition, a variability could arise due to the optimization procedures necessary to learn NN parameters values from data. As the cost function could exhibit many local minima, the NN parameters take on different values entailing variability in estimates. In this case, a parameter uncertainty emerges. Nevertheless, also a model uncertainty could occur for possible structural model misspecification.

Addressing the measurement of uncertainty sources separately is a complex problem, as they are closely connected and no information is available about the input-output relation. However, PIs account for all uncertainty sources, embracing the overall variability around NN point predictions. Therefore, we proceed to define PIs for the LC-LSTM mortality rates in order to estimate the total uncertainty produced by the model integration. 

Recalling that age-dependent parameters are time invariant, the uncertainty in death rates concerns the temporal dynamic described by eq.(\ref{eq:parametric_eq_lstm}). Thus, we focus on the PI construction for the time-index, exploiting the $k_t$ total variance, $\sigma^2_{k_t}$. To this end, the PI characterization is based on the following result. 

\begin{proposition}	
Let $\left( k_t\right)_{t \in \mathcal{T'}}$ the time-index series over the forecast horizon $\mathcal{T'}$. The total variance associated to the time-index value is:
\begin{equation}
	\label{eq:varianza_kt}
	\sigma^2_{k_t} =\sigma^2_{\hat{k}_t}+\sigma^2_{\gamma}+\mathbb{E}\left[BIAS\left(\hat{k}_t \big| \bm{\kappa}_{\mathcal{T'}}\right)^2\right],
\end{equation}

where $BIAS\left(\hat{k}_t \big| \bm{\kappa}_{\mathcal{T'}}\right) = \mathbb{E}\left(\varphi\left(\bm{\kappa}_{\mathcal{T'}}\right) - \hat{k}_t \big | \bm{\kappa}_{\mathcal{T'}}\right)$ and $\sigma^2_{\hat{k}_t}$ is the NN output variance.  \vspace{0.2 cm}
\end{proposition}

\begin{proof}
	Recalling eq.(\ref{eq:stima_kt_lstm}), over the forecasting horizon is straightforward noting that \[\mathbb{E}\left(k_t\right)=\mathbb{E}\left[\mathbb{E}\left(k_t \big | \bm{\kappa}_{\mathcal{T'}}\right)\right] = \mathbb{E}\left(\hat{k}_t\right).\] We proceed to define the time-index variance by direct calculation: 
	\begin{equation}
		\label{eq:variance_proof_v1}
		\begin{aligned}
		\sigma^2_{k_t} & = \mathbb{E}\left[\left(k_t-\mathbb{E}(k_t)\right)^2\right] =  \mathbb{E}\left[\left(k_t-\mathbb{E}(k_t)+\hat{k}_t-\hat{k}_t\right)^2\right] = \\
		& = \mathbb{E}\left[\left(k_t-\hat{k}_t\right)^2\right]+\mathbb{E}\left[\left(\hat{k}_t-\mathbb{E}\left(\hat{k}_t\right)\right)^2\right] +2 \mathbb{E}\left[\left(k_t-\hat{k_t}\right)\left(\hat{k}_t-\mathbb{E}(k_t)\right)\right] \\
		\end{aligned}
	\end{equation}

Assuming a reasonable stochastic independence between $\left(k_t-\hat{k}_t\right)$ and $\left(\hat{k}_t-\mathbb{E}(k_t)\right)$, follows that:

\begin{equation}
\label{eq:variance_proof}
\sigma^2_{k_t} = \mathbb{E}\left[\left(k_t-\hat{k}_t\right)^2\right]+ \sigma^2_{\hat{k}_t}
\end{equation}

The term $\mathbb{E}\left[\left(k_t-\hat{k}_t\right)^2\right]$ identifies the mean squared error of prediction associated to $\hat{k}_t$, whose expression can be developed as below: 

	\begin{equation}
		\label{eq:msep}
		\begin{aligned}
		\mathbb{E}\left[\left(k_t-\hat{k}_t\right)^2\right] & = \mathbb{E}\left[\mathbb{E}\left[\left(k_t-\hat{k}_t\right)^2 \bigg | \bm{\kappa}_{\mathcal{T'}}\right]\right] = \mathbb{E}\left[\mathbb{E}\left[\left(\varphi\left(\bm{\kappa}_{\mathcal{T'}}\right)+\gamma_t-\hat{k}_t\right)^2 \bigg | \bm{\kappa}_{\mathcal{T'}} \right]\right] = \\
		& = \mathbb{E}\left[\mathbb{E}\left[\left(\varphi\left(\bm{\kappa}_{\mathcal{T'}}\right)-\hat{k}_t\right)^2 \bigg | \bm{\kappa}_{\mathcal{T'}} \right]\right] + \mathbb{E}\left[\mathbb{E}\left[\gamma_t^2 \bigg | \bm{\kappa}_{\mathcal{T'}} \right]\right] = \\
		& = \mathbb{E}\left[BIAS\left(\hat{k}_t \big| \bm{\kappa}_{\mathcal{T'}}\right)^2\right] + \sigma^2_{\gamma}. \\
		\end{aligned}
	\end{equation}
	
	Substistuting eq.(\ref{eq:msep}) in eq.(\ref{eq:variance_proof}), we have:
	
	\begin{equation}
		\label{eq:variance_proof_final}
		\sigma^2_{k_t} = 
		 \sigma^2_{\hat{k}_t}+\sigma^2_{\gamma}+\mathbb{E}\left[BIAS\left(\hat{k}_t \big | \bm{\kappa}_{\mathcal{T'}}\right)^2\right], 
	\end{equation}
	completing the proof.
\end{proof}

Following eq.(\ref{eq:varianza_kt}), uncertainty in the future mortality behavior is linked to the NN model. The NN ability to approximate data depends on the function composition extent, which is intrinsically related to the learning process. Hence, $\sigma^2_{\hat{k}_t}$ includes fluctuations due to training data and learned weights, as well as from model misspecifications occurencies. In compliance to the bias-variance principle, an expected bias component is present. In fact, both bias and variance contribute to the NN prediction error and the NN model suitability is based on the reduction of both. Finally, the variance $\sigma^2_{\gamma}$ constitutes an irreducible term of uncertainty, since it refers to the random noise component. 

\subsection{Estimating $\sigma^2_{\hat{k}_t}$}
\label{section:variance_boot}

To derive the NN output variance, the conditioned time-index distribution, $\mathbb{P}\left(\hat{k}_t \big | \bm{\kappa}_{\mathcal{T'}}\right)$, should be known. However, it is not available and we could either hypothesize some distribution or extract it from the data grasped. Considering the latter, our approach to estimate the time-index variance refers to the NN ensemble paradigm, based on the jointly use of multiple NNs (\cite{zhou}). Utilizing a bootstrap technique, multiple training data samples are generated in order to develop an empirical distribution, $\hat{\mathbb{P}}\left(\hat{k}_t \big | \bm{\kappa}_{\mathcal{T'}}\right)$, constitutes by different NN point predictions. The final estimates are then obtained aggregating, by average, the various outputs. The latter procedure, namely bootstrap aggregating or bagging (\cite{breiman}), is an ensemble techniques producing an unbiased estimation and favouring an adequate variance measurement. This means that the expected bias in eq.(\ref{eq:varianza_kt}) is seen as a negligible component affecting the time-index variance (see for instance \cite{khosravi2015}). The bagging scheme proposed in the present work is described in the following steps:
\begin{enumerate}[labelindent=20pt,labelwidth=\widthof{\ref{last-item}},label=\arabic*.,itemindent=1em,leftmargin=!]
\item [\textit{Step 1.}] Using the available time-index series $\bm{\kappa}_\mathcal{T}$, we train the LSTM model to obtain the point estimates in eq.(\ref{eq:stima_kt_lstm}) over the forecast horizon $\mathcal{T'}$;
\item [\textit{Step 2.}] We generate $B \in \mathbb{N}$ samples of $\bm{\kappa}_\mathcal{T}$ through a proper bootstrap procedure. In particular, we refer to the bootstrap strategy proposed in \cite{Koissi};
\item [\textit{Step 3.}] For each $b^{th}$ sample, with $b=1,\ldots,B$, we re-optimize the weights of the function composition defined in \textit{Step 1}. In doing so, only the NN weights will change given the new data and the created NNs ensemble will include uncertainty for both training data and parameters; 
\item [\textit{Step 4.}] For each trained NN in \textit{Step 3}, we predict the associate point estimate on $\mathcal{T'}$, producing a bootstrap distribution consisting of $B$ point predictions, i.e.: 
\begin{equation}
\label{eq:estimed_distribution}
	\hat{\mathbb{P}}\left(\hat{k}_t \big | \bm{\kappa}_{\mathcal{T'}}\right)=\left(\hat{k}_t^{(b)} =\hat{f}_{LSTM}\left(\bm{\kappa}_\mathcal{T}^{(b)},\hat{\bm{\mathcal{W}}}^{(b)}\right), b=1,\ldots,B\right);
\end{equation}
\item [\textit{Step 5.}] From the bootstrap distribution $\hat{\mathbb{P}}\left(\hat{k}_t \big | \bm{\kappa}_{\mathcal{T'}}\right)$, we  find the estimates of interest by aggregation. Hence, the bagged estimate of the variance $\sigma^2 _{\hat{k}_t}$ is:
\begin{equation}
\label{eq:bagged_variance}
		\hat{\sigma}^2_{\hat{k}_t}=\frac{1}{B-1}\sum_{b=1}^{B}\left(\hat{f}_{LSTM}\left(\bm{\kappa}_\mathcal{T}^{(b)},\hat{\bm{\mathcal{W}}}^{(b)}\right)-\overline{k}_{t}\right),
\end{equation}
where $\overline{k}_{t}=\frac{1}{B}\sum_{b=1}^{B}\hat{f}_{LSTM}\left(\bm{\kappa}_\mathcal{T}^{(b)},\hat{\bm{\mathcal{W}}}^{(b)}\right)$ is the bagged estimate for the conditional expectation $\mathbb{E}\left(\hat{k}_t \big | \bm{\kappa}_\mathcal{T'}\right)$.
\end{enumerate}
We emphasize that using an ensemble technique for estimating the NN output variance, the expected bias component is irrelevant. Thus, the ensemble technique could associate high uncertainty to the NN predictions, as the bias-variance trade-off states. Howbeit, if the employed bootstrap technique fits the density estimation problem and the trained NN model is robust, then the estimated variance does not induce an explosive prediction intervals behaviour over time.

\subsection{Estimating $\sigma^2_{\gamma}$}
\label{section:variance_noise}
Mortality dynamic incorporates an intrinsic randomness not explained by the network model, as showed by at eq.(\ref{eq:parametric_eq_lstm}). A NN appropriately trained catches the key input-output data schemes, skimming noisy examples. Consequently, the NN model is suitable to produce forecast, avoiding overfitting occurrences. For our purposes, such noise is analysed and predicted. Considering the training set interval $\mathcal{T}$, we deal with the series  $\left(k_t-\hat{k}_t\right)_{t\in \mathcal{T}}$ as a proxy of the unwrapped noise by NN. It helps to evaluate the estimate $\hat{\sigma}^2_{\gamma}$ as the time-index residual uncertainty over $\mathcal{T}$, spreading the random error over the forecast horizon $\mathcal{T'}$ through a random walk representation.  

\bigskip

\section{Performance metrics of forecasting}
\label{sec:metriche}
To assess quantitatively the LC-LSTM projections over the forecast range, we employ performance metrics both for point and interval forecasts. In the former case, the Root Mean Squared Error (henceforth RMSE) is acknowledged as accuracy measure both for the time-index and mortality rates, respectively:
\begin{equation}
\label{eq:rmse_mae}
RMSE_{(k)}=\sqrt{\frac{\sum_{t=t_{n}+1}^{t_{n}+s}\left(k_{t}-\hat{k}_{t}\right)^2}{s-1}}, \quad
RMSE_{(m)}=\sqrt{\frac{\sum_{t=t_{n}+1}^{t_{n}+s}\left(\log m_{x,t}-\log\hat{m}_{x,t}\right)^2}{s-1}}. \\
\end{equation}
To judge PI quality and effectiveness, we jointly examine PI coverage probability and PI width. In analytical terms, we consider two indicators namely the Prediction Interval Coverage Probability (henceforth PICP) and the Mean Prediction Interval Width (henceforth MPIW). The former inspects the PI coverage counting how many values are wrapped in the probabilistic range, given a confidence level. In other words, the PICP estimates the probability that the mortality rates values fall within the PI provided by the mortality model. Let $\hat{k}_{t}^{L}$ be the estimated time-index lower bound and be $\hat{k}_{t}^{U}$ the estimated time-index upper bound. Then, the PICP for the $k_t$ series is defined as follows:
\begin{equation}
PICP_{(k)}= \frac{1}{s-1}\sum_{t=t_{n}+1}^{t_{n}+s} \bm{1}_{\left\lbrace \hat{k}_{t} \ \in \ [\hat{k}_{t}^{L}, \hat{k}_{t}^{U}]\right\rbrace},
\end{equation} 
where $\bm{1}_{\left\lbrace \cdot \right\rbrace}$ is the indicator function such that $\bm{1}_{\left\lbrace \cdot \right\rbrace}= 1$ if $ \hat{k}_{t} \in [\hat{k}_{t}^{L},\hat{k}_{t}^{U}]$, and $\bm{1}_{\left\lbrace \cdot \right\rbrace}=0$ otherwise. 

The MPIW indicates the PI mean width over forecasting horizon, that is:

\begin{equation}
MPIW_{(k)} = \frac{1}{s-1} \sum_{t=t_{n}+1}^{t_{n}+s} \hat{k}_{t}^{U}-\hat{k}_{t}^{L}.
\end{equation}

We also calculate PICP and MPIW on the log-mortality rates by a given age $x$. Let $\log \hat{m}_{x,t}^{L}$ be the estimated mortality rates lower bound and be $\log \hat{m}_{x,t}^{U}$ the estimated mortality rates upper bound. Then, we specify the PICP and MPIW as follows:
\begin{equation}
PICP_{(m)}= \frac{1}{s-1}\sum_{t=t_{n}+1}^{t_{n}+s} \bm{1}_{\left\lbrace \log\hat{m}_{x,t} \ \in \ [\log\hat{m}_{x,t}^{L}, \log\hat{m}_{x,t}^{U}]\right\rbrace},
\end{equation} 
where $\bm{1}_{\left\lbrace \cdot \right\rbrace}= 1$ if $ \log \log\hat{m}_{x,t} \in [\log\hat{m}_{x,t}^{L},\log\hat{m}_{x,t}^{U}]$, and $\bm{1}_{\left\lbrace \cdot \right\rbrace}=0$ otherwise, and 

\begin{equation}
MPIW_{(m)} = \frac{1}{s-1} \sum_{t=t_{n}+1}^{t_{n}+s} \log\hat{m}_{x,t}^{U}-\log\hat{m}_{x,t}^{L}.
\end{equation}

Higher PICP value indicates PIs having a greater probability to cover the true mortality realizations. High MPIW values are desirable in order to provide a suitable uncertainty portrayal. An explosive demeanour in variability is reflected by greater MPIW levels, jeopardizing the biological plausibility of mortality forecasts. The latter qualitative criterion is valuable, since it concerns the predicted uncertainty levels consistency  w.r.t. the historical volatility at different ages (\cite{Cairns2011}). 

\section{Empirical investigation and results}
\label{sec:numerical_application}
In the following we illustrate the empirical analysis carried out to test our model proposal. The results and considerations presented will also take into account the forecasts getting from the LC Poisson model (\cite{Brouhns}) as a term of comparison. The equations defining the LC Poisson predictions are reported in Appendix A. Our analysis has been achieved using the R software (version 3.6.3), exploting the packages \textit{StMoMo} (version 0.4.1), \textit{forecast} (version 8.13), \textit{Keras} (version 2.2.5) and \textit{Tensorflow} (version 1.13.1). 

\subsection{Data}
Our numerical experiment concerns three countries worldwide, Australia, Japan and Spain, analyzed by gender. Data were downloaded from the Human Mortality Database \citep{HMD}, considering the age set $\mathcal{X}=\left\lbrace0,1,\ldots,99 \right\rbrace$. We consider two calendar years sets, 1950-2018 and 1960-2018, to assess both accuracy and variability of the LC-LSTM outcomes with respect to the historical time chunks. This allows us to verify the shortening of the NN training set effects on the learning process, i.e. the network robustness to changes in the training set length. 
For each country, period and gender under investigation, the mortality data are processed through the elaborations explained in the following Section.

\subsection{Neural Network tuning, training and ensembling}
To apply the LSTM model, firstly we fit the LC model in eq.(\ref{eq:predictor_LC}) to the observed age-period mortality data, estimating the age-dependent parameters and the time-index series $\left( k_t\right)_{t\in \mathcal{T}}$. We pose $j=1$ to define the one-step lagged time series, i.e. $\bm{\kappa}_\mathcal{T}=(k_{t-1})_{t\in \mathcal{T}}$, imposing to the LSTM model to sift mortality data at annual paces, that is  $k_t=f_{LSTM}\left(k_{t-1}; \bm{\mathcal{W}}\right)+ \gamma_t$ according to eq.(\ref{eq:parametric_eq_lstm}).   

To tune and train the NN model is necessary to split the time-index series into distinct datasets. To this end, we exploit a hierarchical procedure. Setting $T=2000$ as forecasting year for all the countries investigated, we define the training set and the testing set as below:  
\begin{equation}
\label{eq:train_test}
\begin{aligned}
\text{TRAINING SET:} & \quad \mathcal{TR}=\left(k_t | k_{t-1} \right)_{t=t_0,\ldots,T}\\
\text{TESTING SET:}& \quad \mathcal{TS}=\left(k_t | k_{t-1} \right)_{t=T+1,\ldots,t_n}, \\
\end{aligned}
\end{equation}
where $t_0=\left\lbrace1950, 1960\right\rbrace$ and $t_n=2018$. In addition, to validate the model we divide the training set into a sub-training set and in a validation set, considering the splitting rule $80\%-20\%$. Hence, denoting with $T^{\text{sub}}$ the last year in the sub-training set, we have:
\begin{equation}
\label{eq:train2_val}
\begin{aligned}
\text{SUB-TRAINING SET:}& \quad \mathcal{TR}^{\text{sub}}=\left(k_t | k_{t-1}\right)_{t=t_0,\ldots,T^{\text{sub}}} \\
\text{VALIDATION SET:}& \quad \mathcal{VS}=\left(k_t | k_{t-1}\right)_{t=T^{\text{sub}}+1,\ldots,T} \\
\end{aligned}
\end{equation}

We use the sets $\mathcal{TR}^{\text{sub}}$ and $\mathcal{VS}$ to tune the NN structure through a grid search technique. Thus, a bounded discrete parametric space is a priori settled, whose possible values are arbitrarily chosen acting as network hyper-parameters. Fixing a hyper-parameters combination, the learning process begins minimizing the Mean Squared Error cost function over the set $\mathcal{TR}^{\text{sub}}$. We select as optimal NN structure the one identified by the hyper-parameters combination returning the minimum error on the validation set $\mathcal{VS}$. In doing so, the function composition, $\hat{f}_{LSTM}$, is built according to the data. For each countries and both genders under investigation, the LSTM model is characterized by $p=1$ hidden layer, considering the ReLu function (\cite{nair}) as feedforward activation function, the tangent hyperbolic function as recurrent activation function and the linear function as the output layer activation function $\psi$. The number $N_p$ of hidden neurons varies depending on both countries and genders studied. Finally, the best NN architecture is afterwards employed on the training set, $\mathcal{TR}$, to spawn point predictions over the testing set horizon. Therefore, we compare the NN forecasts, $\hat{k}_t$, with the available time-index values in $\mathcal{TS}$ as backtesting exercise.   

The depicted learning process suggests the minimum learning period length to produce robust predictions. Shortening the training dataset, our experiment highlights that training periods beginning after 1960s generate predictions sensitive to small variations in the data. Therefore, we need at least 40 observations to adequately tune the network model. 

The tuned LSTM model acts as the reference model in \textit{Step 1.} of the proposed bagging scheme in Section \ref{section:variance_boot}. Following the bootstrap strategy proposed in \cite{Koissi}, we generate $B=1000$ bootstrap samples of the training set $\mathcal{TR}$. Maintaining the tuned network function composition, $\hat{f}_{LSTM}$, we estimates its weights on the $b^{th}$ training set producing the related forecasts over testing set horizon. Therefore, the bootstrap distribution $\hat{\mathbb{P}}\left(\hat{k}_t \big | k_{t-1}\right)$ is obtained, allowing for the bagged variance calculation as in eq.(\ref{eq:bagged_variance}).

\subsection{Results}
\label{sec:results}
In the following we provide the results of our numerical application, recalling the performance metrics exposed in Section \ref{sec:metriche}. We firstly refer to the RMSE metric to judge the point forecasts accuracy, considering also the error of the LC projections as benchmark. To appreciate PIs quality by PICP and MPIW indicators, after the bagging scheme we need to assess the noise variance in order to estimate PI boundaries. We consider the sample variance of the series $\left(k_t-\hat{k}_t\right)_{t \in \mathcal{TR}}$ as the noise variance estimate over training set. To project the noise and its uncertainty over testing set horizon, we inspect its possible random walk behaviour. To this end, the Augmented Dickey Fuller (ADF) test is implemented. In addition, we test normality features of the noise realizations through statistical normality tests, such as the Shapiro-Wilk, the D'Agostino-Pearson and the Jarque-Bera. For all countries and both genders investigated, the achieved noise analysis confirms the suitability of a random walk representation with Gaussian innovations for the noise component (see Appendix B). Therefore, the LC-LSTM time-index values are embedded within the following PI, for a confidence level $\alpha$:
\begin{equation}
	\left[\hat{k}_t^L, \hat{k}_t^U\right]=\left[\hat{k}_t - z_{\frac{\alpha}{2}}\sqrt{\hat{\sigma}^2_{\hat{k}_t}+\hat{\sigma}^2_{\gamma}} \ , \ \hat{k}_t + z_{\frac{\alpha}{2}}\sqrt{\hat{\sigma}^2_{\hat{k}_t}+\hat{\sigma}^2_{\gamma}}\right]
\end{equation}       
where $z_{\alpha}$ is the $\alpha$-quantile of a Standard Normal distribution. 

We proceed to calculate performance metrics for the LC-LSTM, as well as for the LC model. Their values for the time-index appear in \autoref{tab:result}, comparing the LSTM performances, for the LC-LSTM, with the ARIMA ones, for the LC model.

\begin{table}[H]
	\centering
	\renewcommand*{\arraystretch}{1.15}
	\caption{$k_t$ performance metrics values for each training period. Forecasting years: 2001-2018.\label{tab:result}}
	\begin{adjustbox}{max width=\textwidth}
		\begin{tabular}{|c|c|ccc|ccc|ccc|ccc|}
			\hline
			\multirow{3}{*}{\textbf{Country}} & \multirow{3}{*}{\textbf{Model}} & \multicolumn{6}{c|}{\textbf{Training period 1950-2000}} & \multicolumn{6}{c|}{\textbf{Training period 1960-2000}} \tabularnewline
			\cline{3-14}
			& & \multicolumn{3}{c|}{\textbf{Male}} & \multicolumn{3}{c|}{\textbf{Female}} & \multicolumn{3}{c|}{\textbf{Male}} & \multicolumn{3}{c|}{\textbf{Female}} \tabularnewline
			\cline{3-14}
			& &  $RMSE$ & $PICP_{(k)}$ & $MPIW_{(k)}$ & $RMSE$ & $PICP_{(k)}$ & $MPIW_{(k)}$ & $RMSE$ & $PICP_{(k)}$ & $MPIW_{(k)}$ & $RMSE$ & $PICP_{(k)}$ & $MPIW_{(k)}$ \tabularnewline
			\hline	
			\multirow{2}{*}{\textbf{\textit{Australia}}} & \multirow{1}{*}{ARIMA} &	9.514 & 1 & \textbf{53.503} & 3.861 &	1 &	25.195 & 5.138 & 1 & \textbf{47.485} & 3.637 &	1	& 25.089 \tabularnewline 
			& \multirow{1}{*}{LSTM} & \textbf{4.280} & 1 & 32.865 & \textbf{3.790}  & 1	& \textbf{39.478} & \textbf{1.970} & 1 & 28.143 & \textbf{2.659} & 1 & \textbf{37.433} \tabularnewline 
			\hline		
			\multirow{2}{*}{\textbf{\textit{Japan}}} & \multirow{1}{*}{ARIMA} & 3.743 & 1 & 21.503 & \textbf{10.084} & 0.556 & 20.767 & 4.647 &  1 & 17.392 & 9.790 & 0.500 & 12.409 \tabularnewline 
			& \multirow{1}{*}{LSTM} & \textbf{2.228} & 1 & \textbf{43.784} & 18.014 & \textbf{1} & \textbf{53.431} & \textbf{2.069} & 1 & \textbf{28.209}
			& \textbf{5.818} & \textbf{1} & \textbf{30.701} \tabularnewline 
			\hline
			\multirow{2}{*}{\textbf{\textit{Spain}}} & \multirow{1}{*}{ARIMA} & 14.038 & 0.333 & 19.354 & \textbf{6.215} & 1 & 21.394 & 13.071 & 0.333 & 17.343 & 5.805 & 1 & 20.747 \tabularnewline 
			& \multirow{1}{*}{LSTM} & \textbf{8.625} & \textbf{1} & \textbf{35.424} & 7.471 & 1 & \textbf{60.373} & \textbf{9.983} & \textbf{0.778}	& \textbf{23.340} & \textbf{4.357} & 1 & \textbf{28.141} \tabularnewline 
			\hline
		\end{tabular}
	\end{adjustbox}
\end{table}

For all the countries considered, the time-index series experienced since the 1960s exhibit a markable linear decline over time. In particular, mortality reductions accelerated over period 1950-1960, and an approximately constant rate of degrowth characterize the interval 1960-2000. Such a behaviour has been driven by a decline in infant mortality, as well as reductions in mortality at older ages after the WWII (see for instance \cite{rau}). 

As general statement about predictions accuracy, our analysis confirm the ARIMA process suitability to represent linear evolution in mortality. On the other side, the LSTM seems to be advisable for linear, noisy or non-linear series. Scrutinizing the uncertainty results, the LSTM offers always a greater probability coverage, in most cases due to the PI width. Because of the LSTM point predictions present low bias, their variance tend to be increasing and to be higher than the ARIMA one.

The majority of cases promote the LSTM model usefulness in affording a more actual mortality trend, as well as for uncertainty estimation. The most virtuous example concerns the Australian males, presenting the lower RMSE on the period 1960-2000. Considering the training period 1950-2000, the NN allows the simultaneous presence of a total coverage of the future $k_t$ realizations and a proper PI width. This situation appears also reducing the training set length, i.e. considering the interval 1960-2000.  
A suitable mortality dynamic for the ARIMA model is offered by Japanese females. In fact, their mortality behaviour presents a strong linear decreasing over time, also if observed from 1950. In this circumstance, the LSTM learns a too steep trend of mortality reductions, as opposed to ARIMA. However, switching to the training period 1960-2000 the network performances improve significantly. We observe a gain of 67,7\% in RMSE terms, maintaining at the same time both a total probability coverage and a coherent MPIW value. On the other side, the ARIMA model does not favour a reliable uncertainty estimation in both periods. Its coverage probability is around 50\%, indicating that the predictive model fails, on average, to anticipate half of the future realizations. An analogous result holds for the Spanish males, whose time-index dynamic shows a noisier series over both training periods. Indeed, the ARIMA coverage probability for Spanish males remains stable around 33\%.           

Furthermore, we also depict the mortality profile for both genders considering ages 45, 65 and 85. To explore such a results, we display the computed performance metrics in \autoref{tab:result_rates}, as well as the PIs graphs in \autoref{fig:pi_males} and \autoref{fig:pi_females}.

\begin{figure}[H]
	\centering
	\caption{MALE PI ($\alpha=5\%$). Forecasting period: 2001-2018. Training period: 1950-2000 (left), 1960-2000 (right).}
	\label{fig:pi_males}
	\includegraphics[width=0.45\linewidth,height=0.25\textheight]{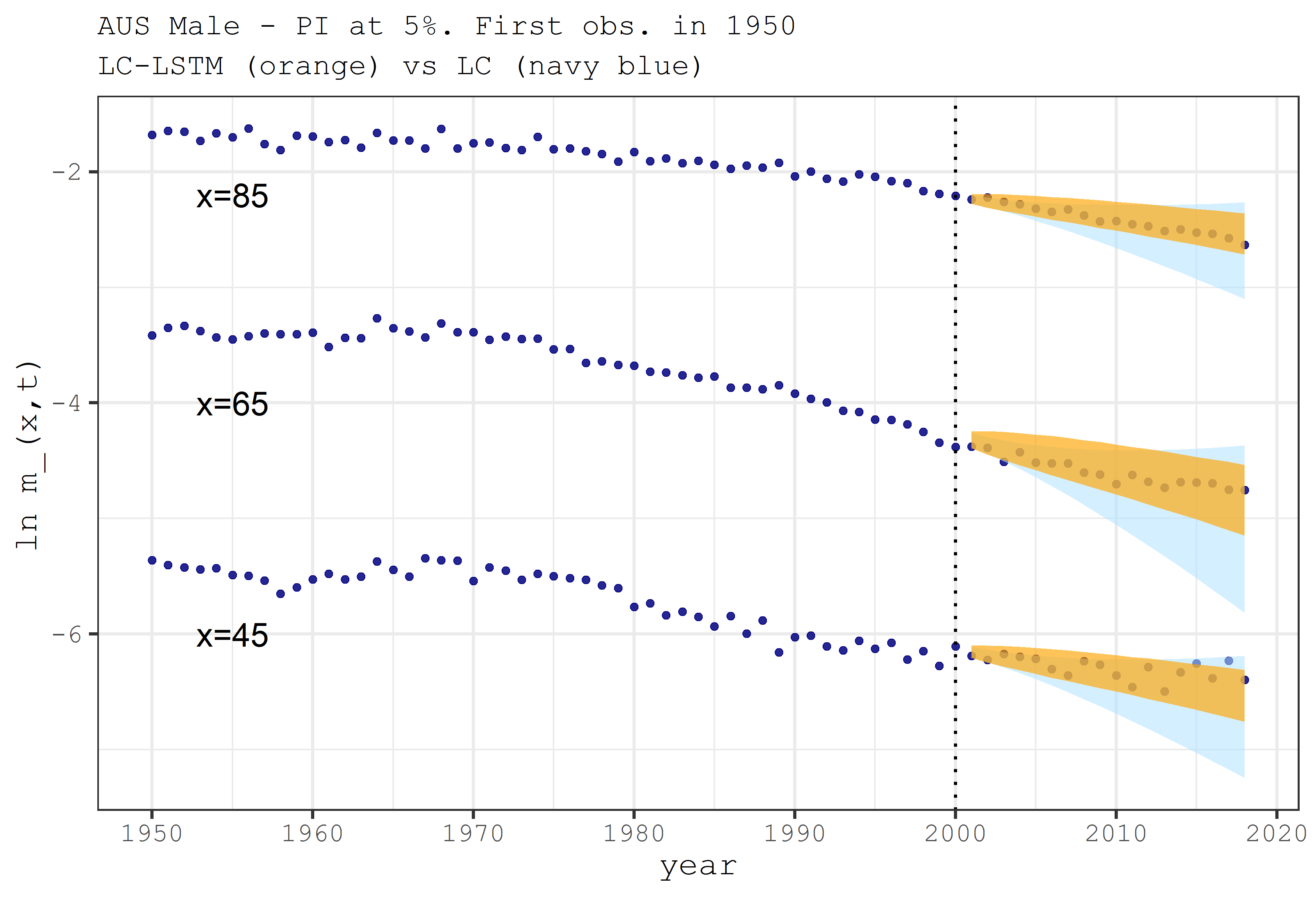}  \quad
	\includegraphics[width=0.45\linewidth,height=0.25\textheight]{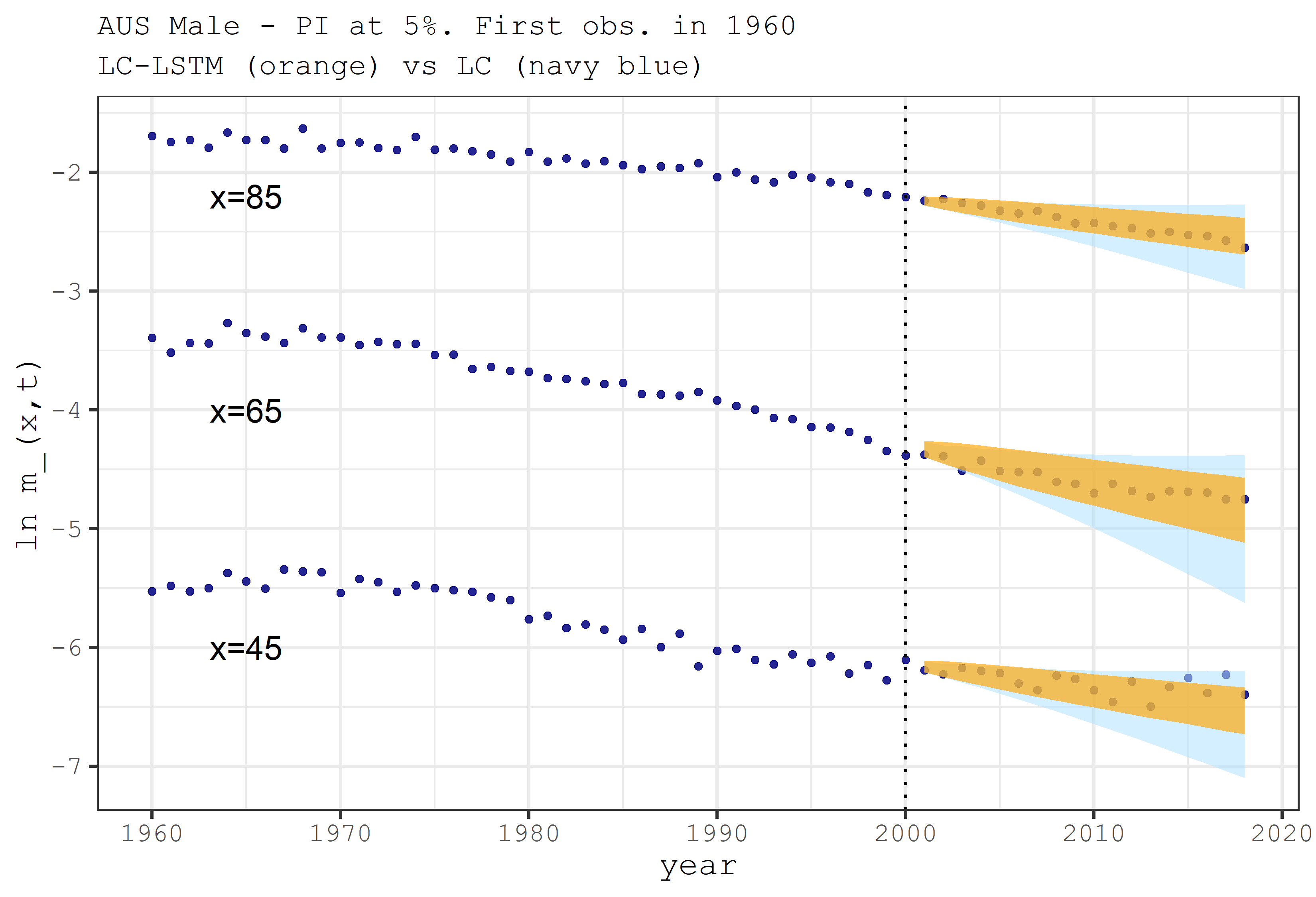} \\
	\includegraphics[width=0.45\linewidth,height=0.25\textheight]{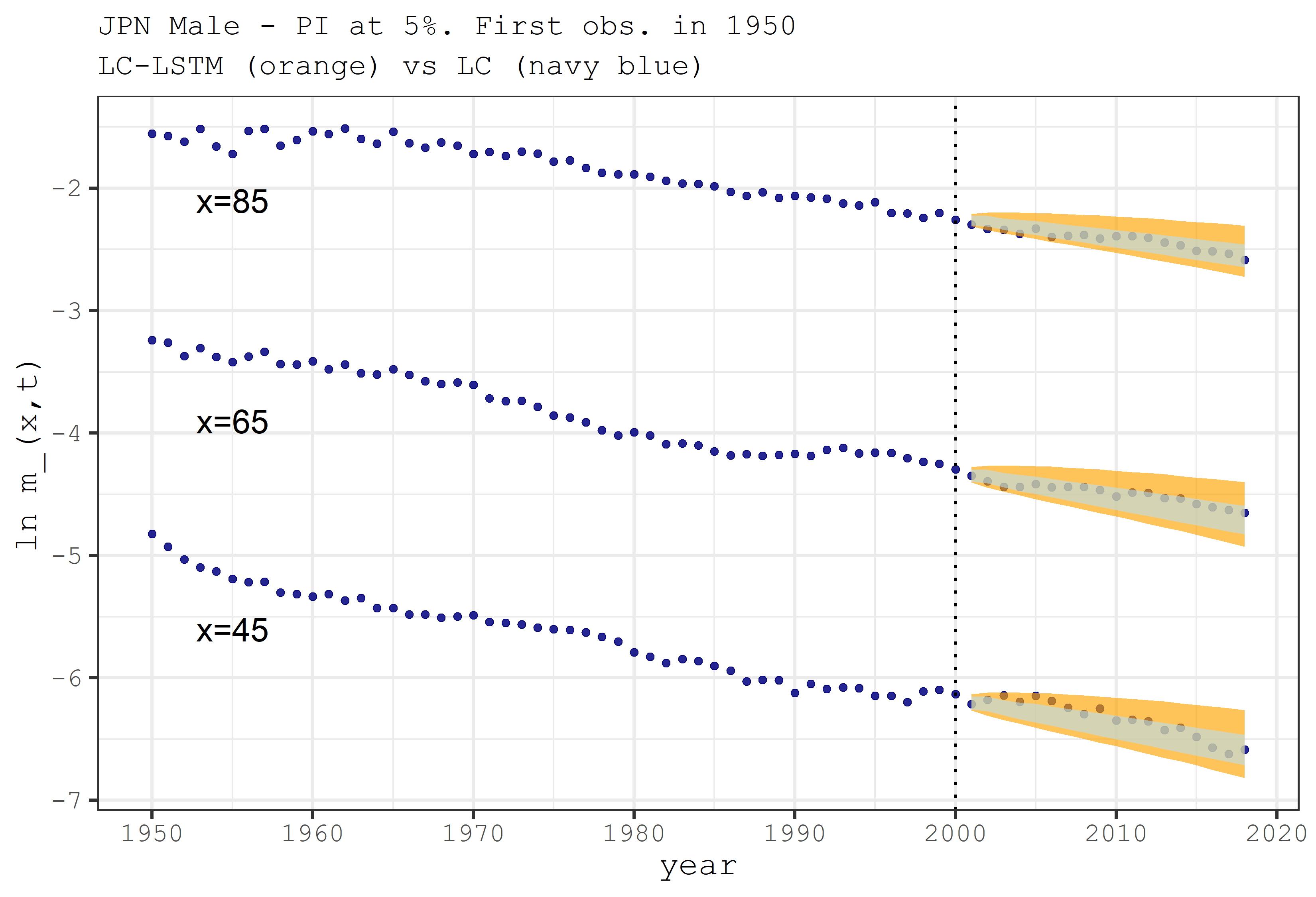}  \quad
	\includegraphics[width=0.45\linewidth,height=0.25\textheight]{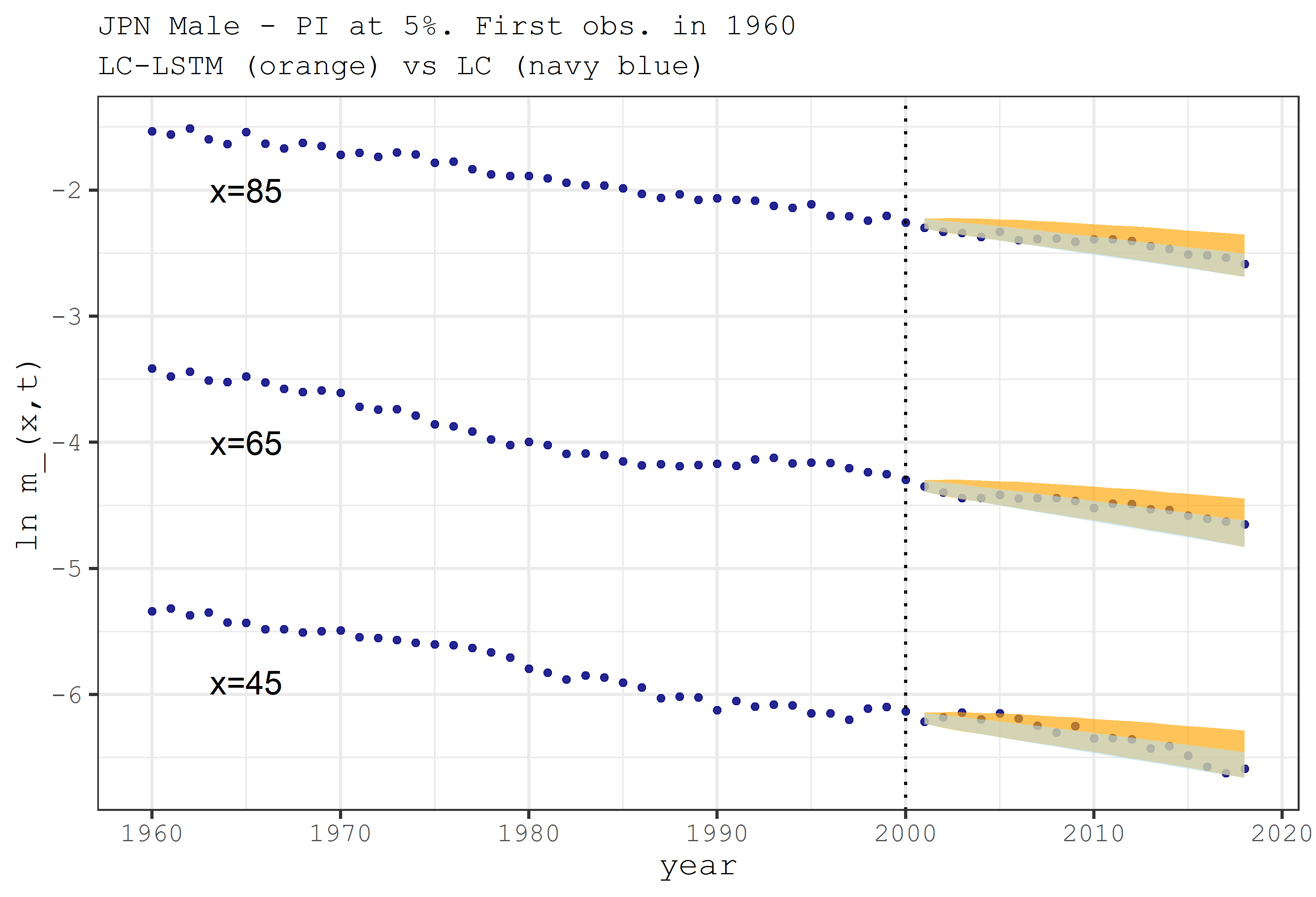} \\
	\includegraphics[width=0.45\linewidth,height=0.25\textheight]{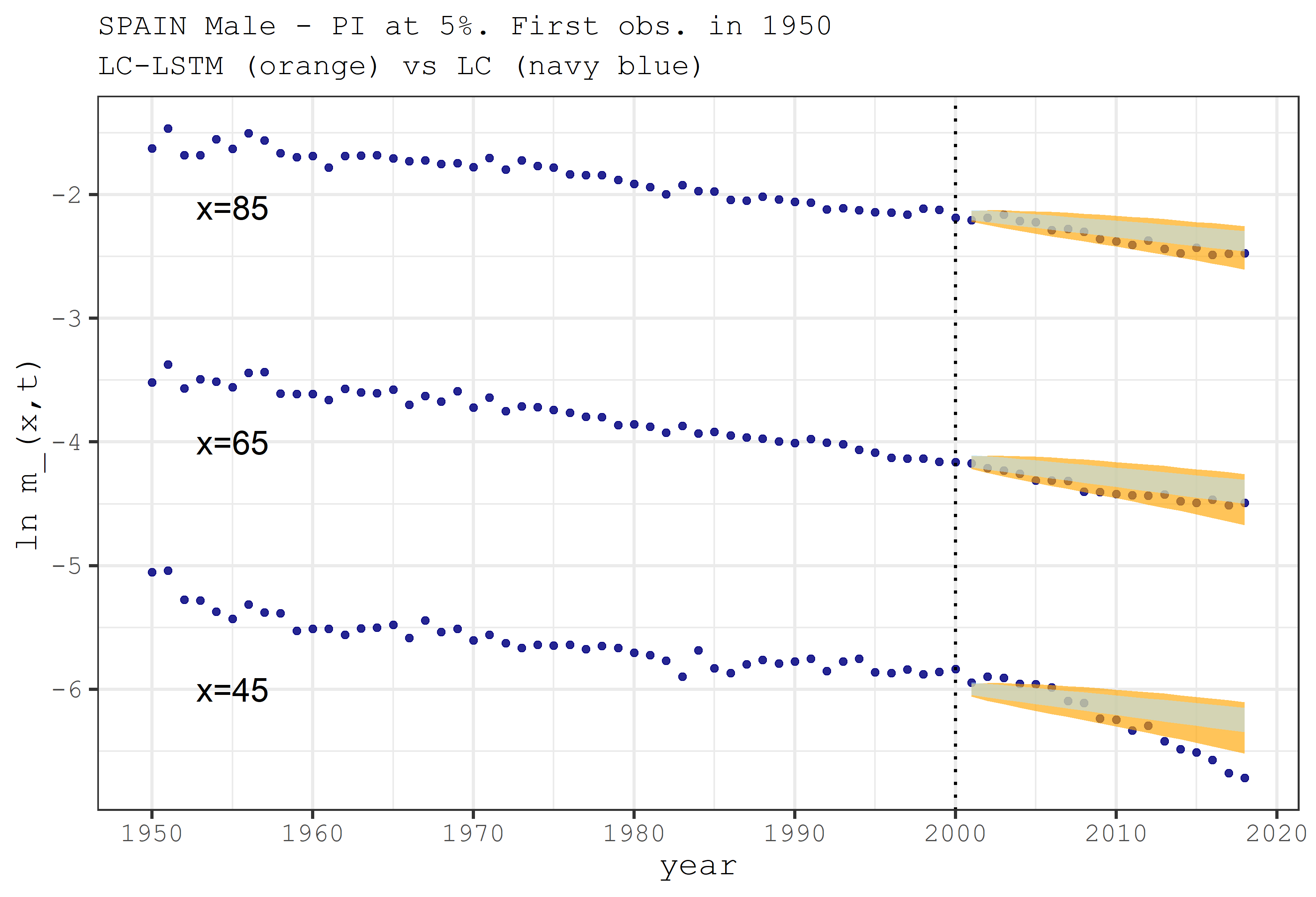}  \quad
	\includegraphics[width=0.45\linewidth,height=0.25\textheight]{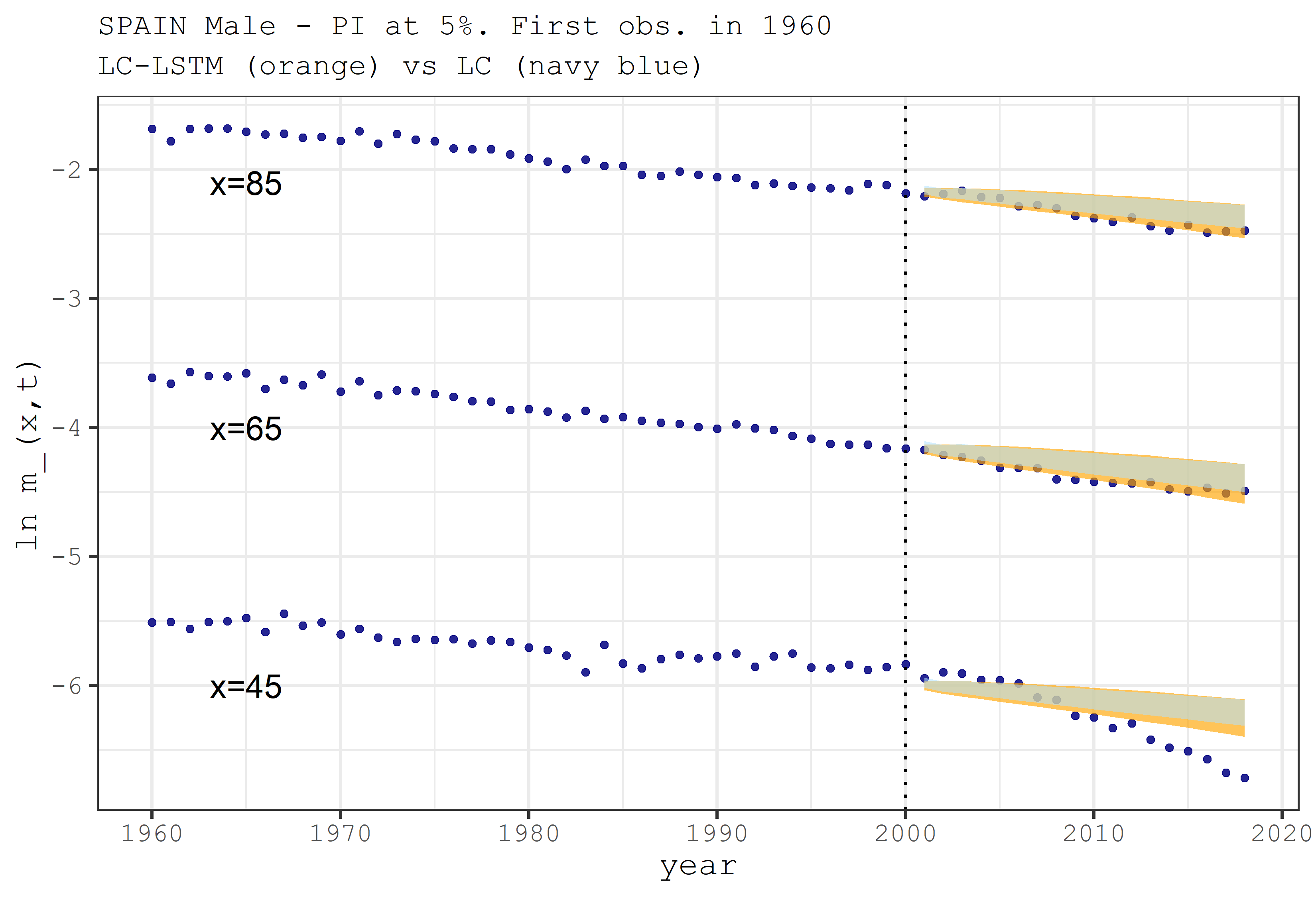} \\
\end{figure}

We can highlight the estimated PIs for the LC-LSTM model both in terms of point and interval estimates. Looking at the Japanese population, we endorse the findings in \autoref{tab:result} for ages 45 and 65. The LC-LSTM provides boundaries proper shaped according to death rates, while the LC model presents the narrowest ranges of variability lacking uncertainty information. For example, over training period 1960-2000 for the Japanese females aged 65, the PIs from LC model shows a coverage probability around $33\%$, while the LC-LSTM provides $PICP_{(m)}=1$ with a similar interval width. For the age 85, whose mortality reductions present slower linear changes over time, also the LC fits the future mortality profile. 

\begin{figure}[!h]
	\centering
	\caption{FEMALE PI ($\alpha=5\%$). Forecasting period: 2001-2018. Training period: 1950-2000 (left), 1960-2000 (right).}
	\label{fig:pi_females}
	\includegraphics[width=0.45\linewidth,height=0.25\textheight]{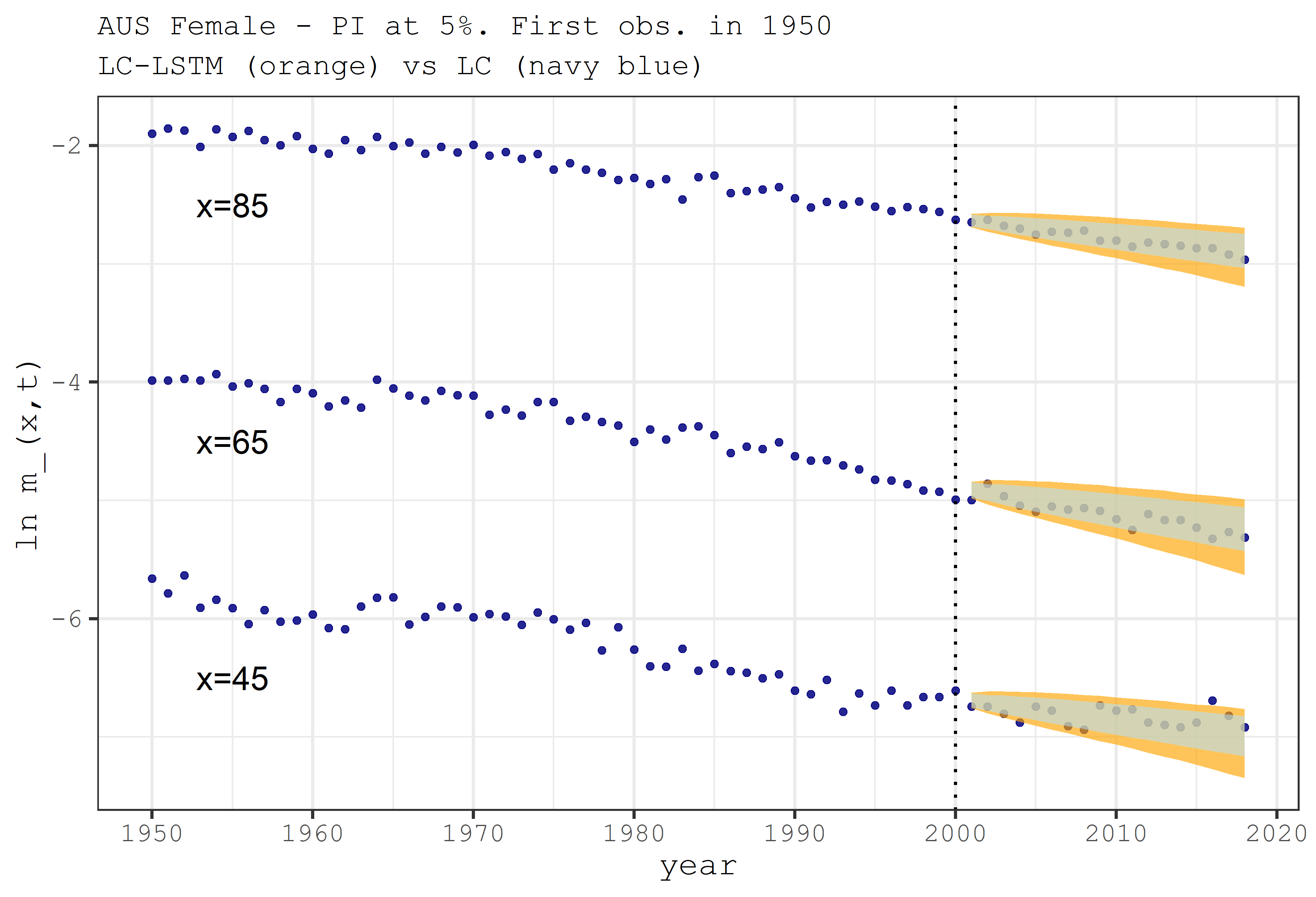}  \quad
	\includegraphics[width=0.45\linewidth,height=0.25\textheight]{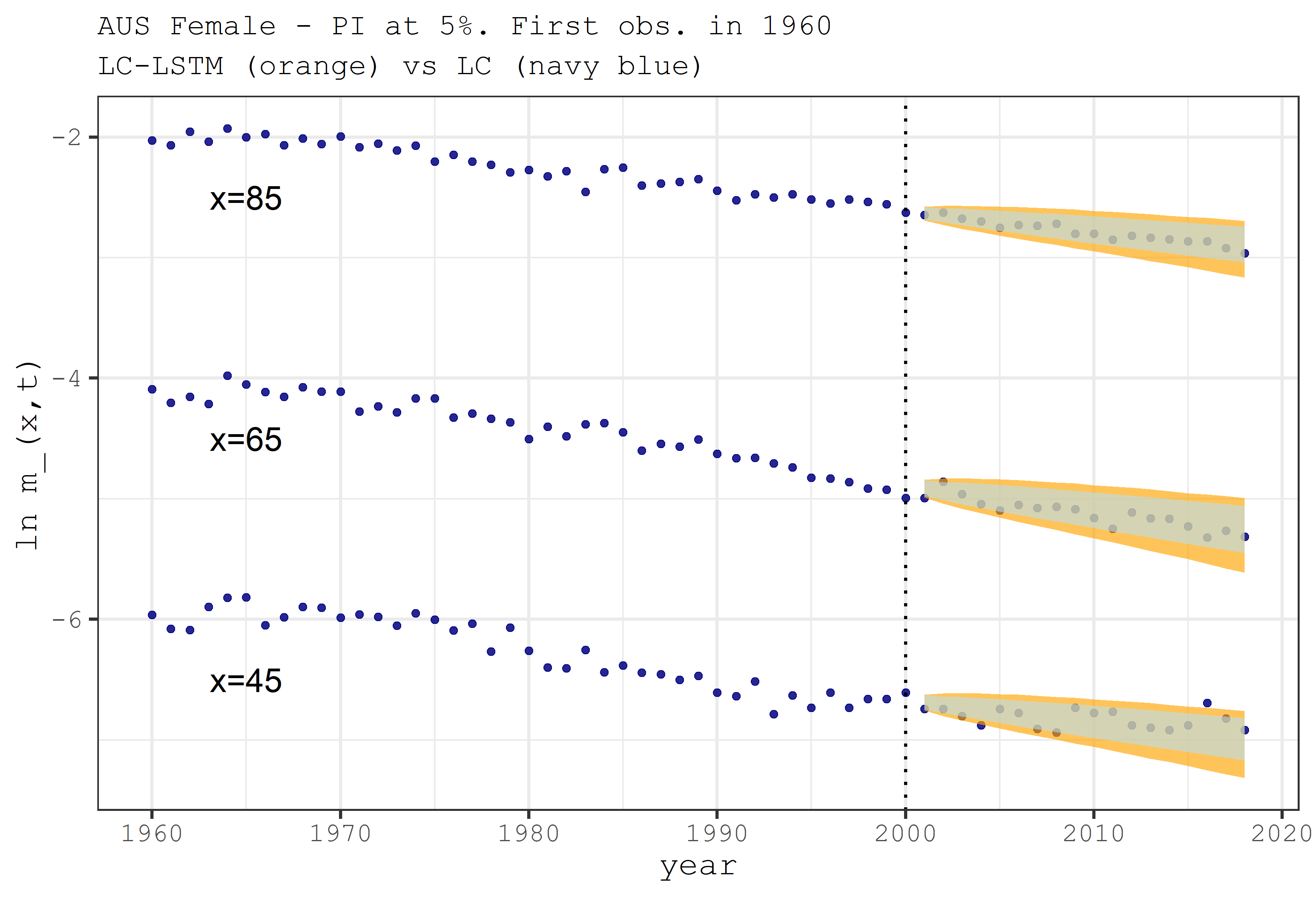} \\
	\includegraphics[width=0.45\linewidth,height=0.25\textheight]{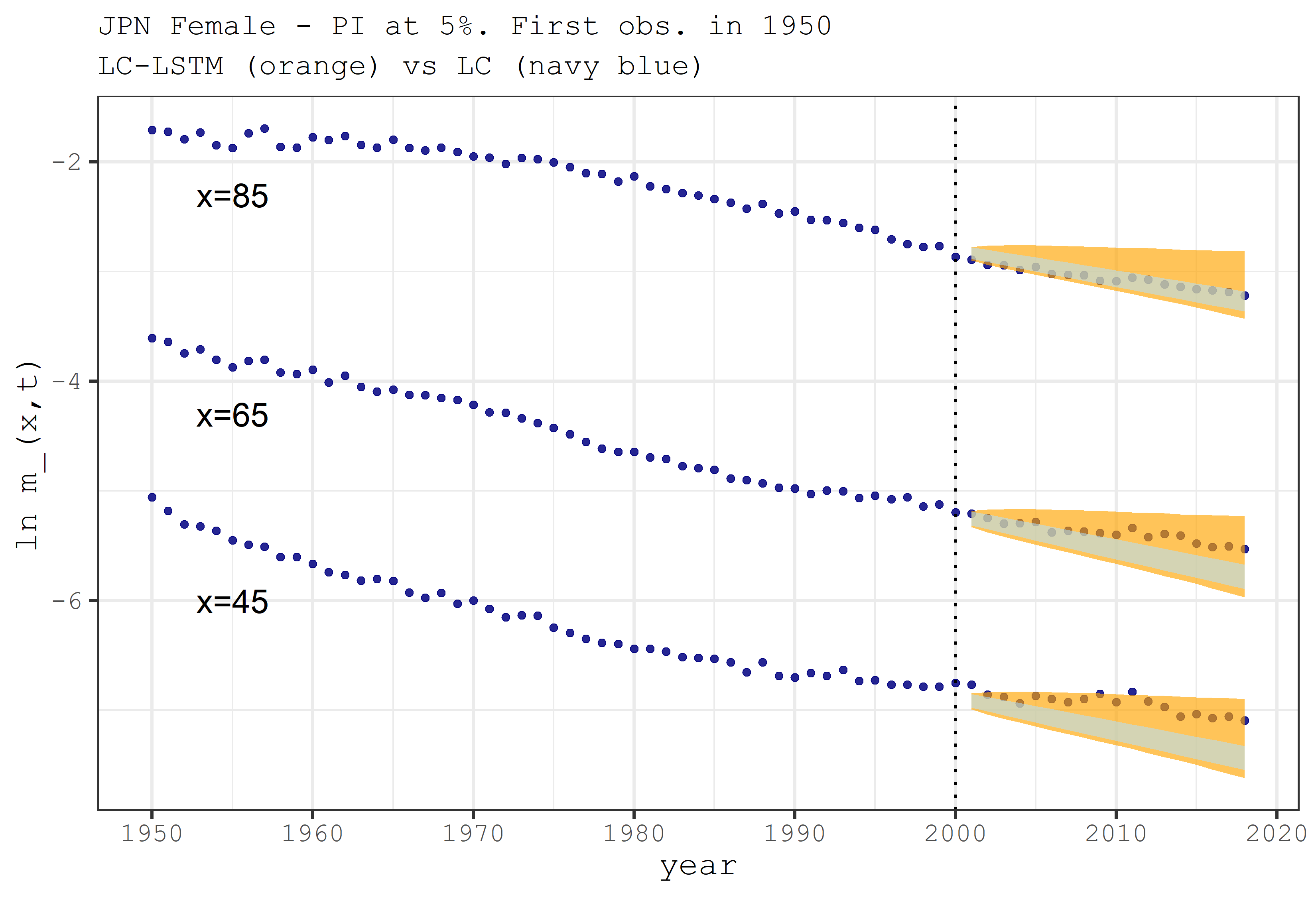}  \quad
	\includegraphics[width=0.45\linewidth,height=0.25\textheight]{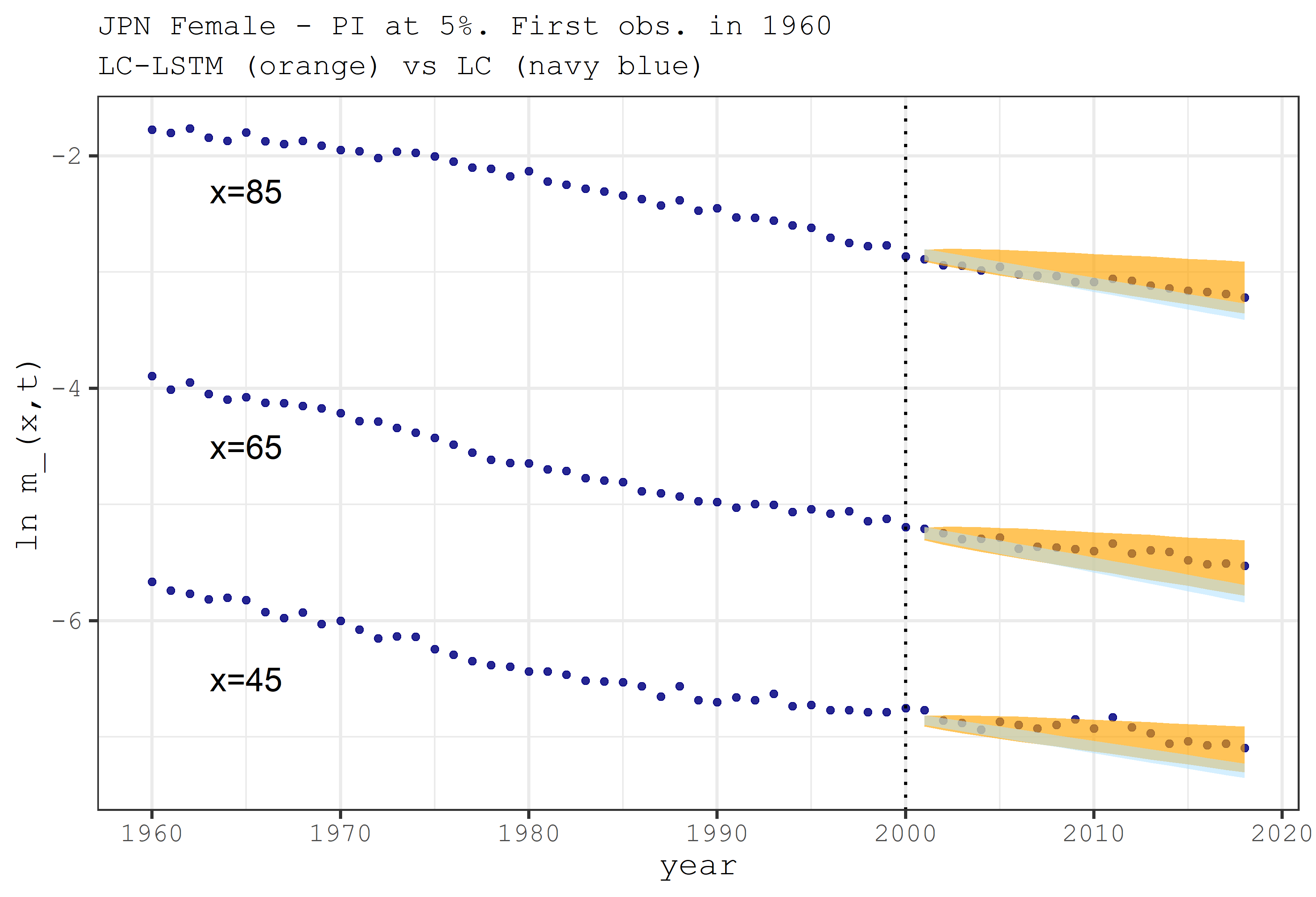} \\
	\includegraphics[width=0.45\linewidth,height=0.25\textheight]{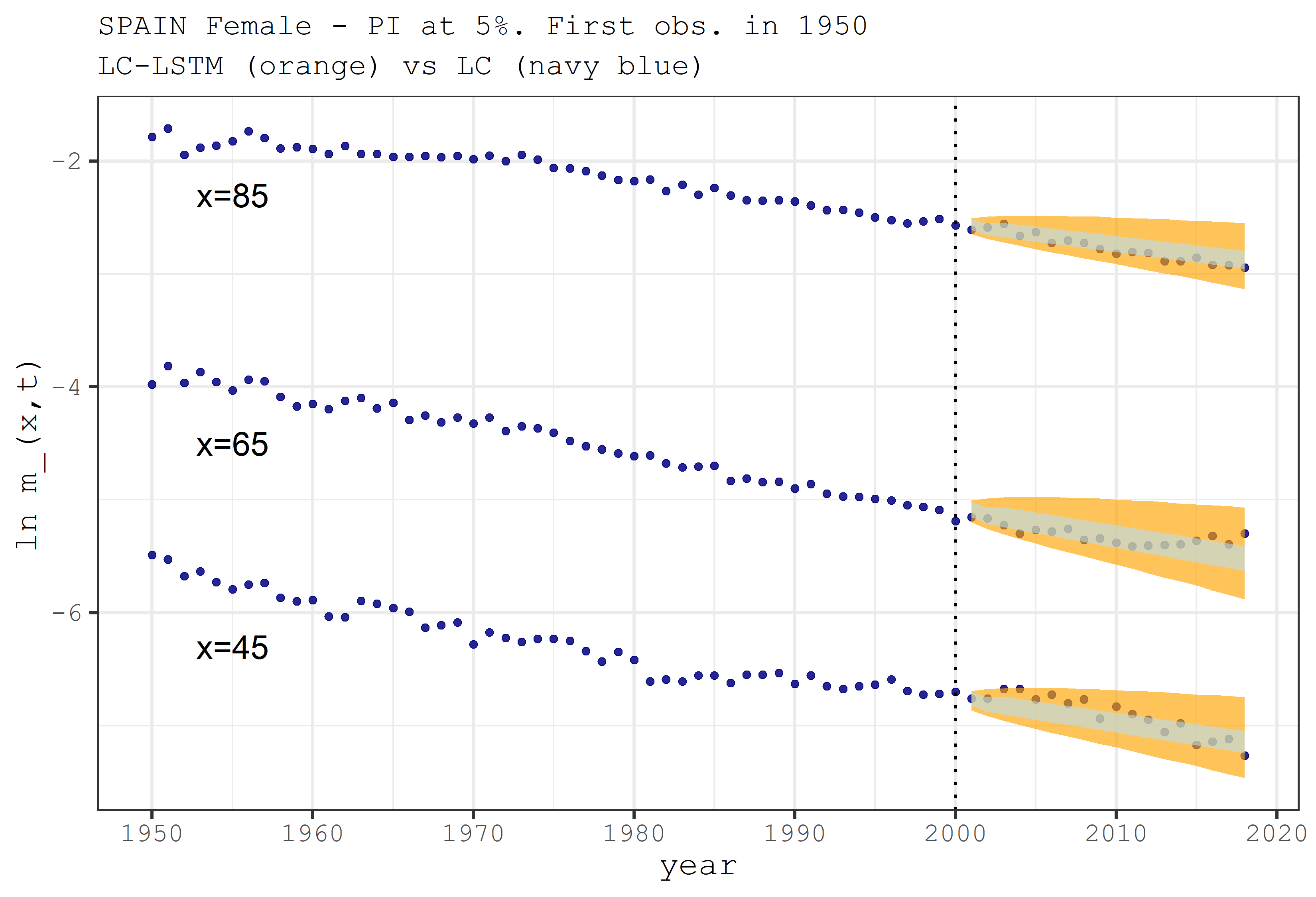}  \quad
	\includegraphics[width=0.45\linewidth,height=0.25\textheight]{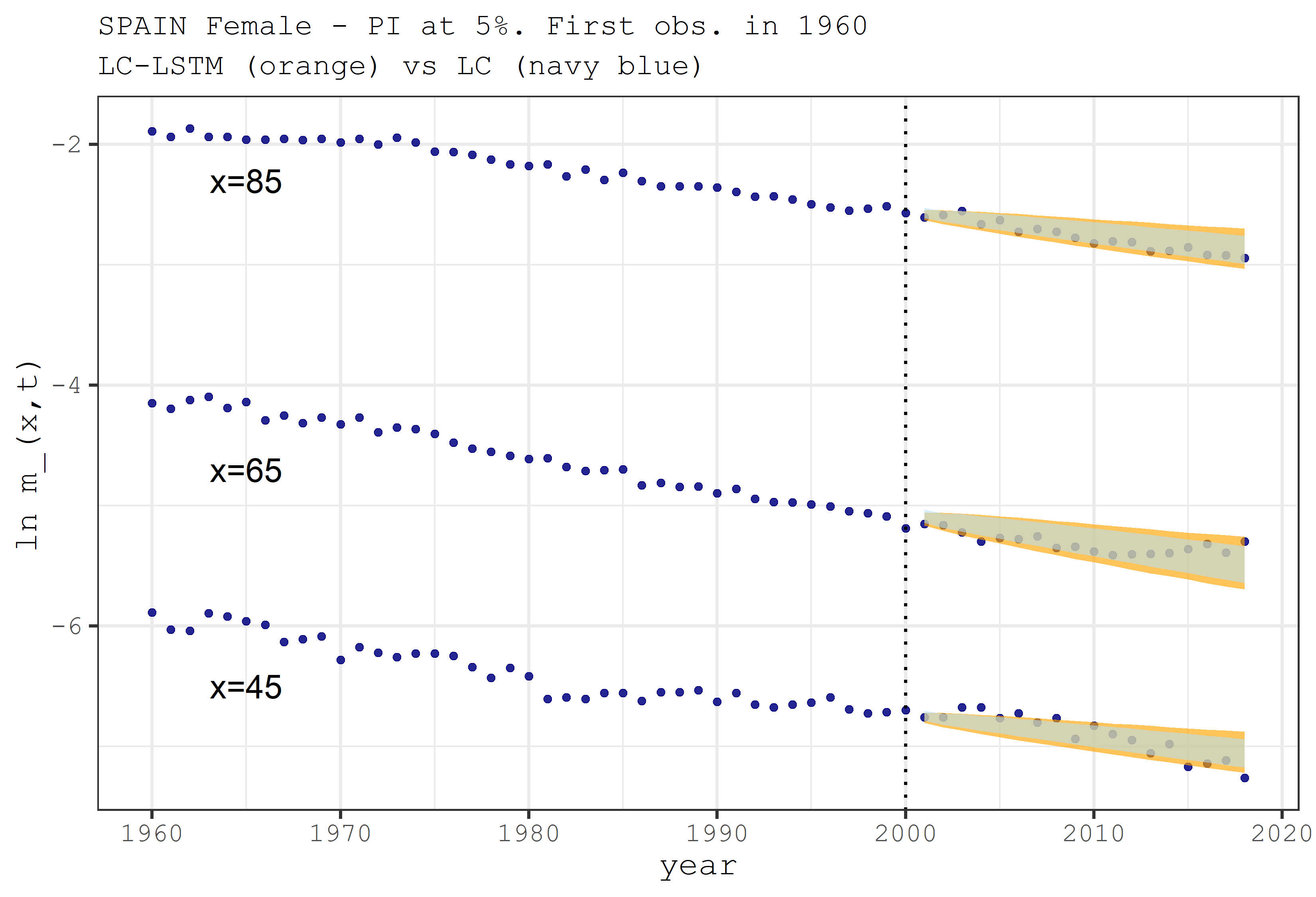} \\
\end{figure}

For Spanish population the LC-LCTM seems to be the befitting model for predictive purposes. As reported in \autoref{tab:result} for both Spanish genders, as the training period shifts the MPIW indicator for the time-index identifies a significant reductions in the PI width ($-20.56\%$ for males and $-53.38\%$ for females), although full probability coverage is maintained. Such a reduction affects the uncertainty measurement from the LC-LSTM model, albeit with an ever wider PI than the LC model one. Finally, we stress how both the LC model and the LC-LSTM model fail to catch the non-linear mortality pattern characterizing age 45 over testing horizon. Starting from the 2000s, Spanish males aged 45 has been experimented a notable acceleration in the rate of mortality reduction. Since we pose $T=2000$ as the forecasting year, the extrapolation approach underlying both the LC and the LC-LSTM induces misleading projections.

\begin{table}[H]
	\centering
	\renewcommand*{\arraystretch}{1.2}
	\caption{$\log m_{x,t}$ performance metrics values for each training period. Forecasting years: 2001-2018.\label{tab:result_rates}}
	\begin{adjustbox}{max width=\textwidth}
		\begin{tabular}{|c|c|ccc|ccc|ccc|ccc|}
			\multicolumn{1}{c}{} & \multicolumn{1}{c}{} & \multicolumn{11}{c}{\LARGE$\bm{x=45}$} & \multicolumn{1}{c}{ } \tabularnewline
			\cline{1-14}
			\hline
			\multirow{3}{*}{\textbf{Country}} & \multirow{3}{*}{\textbf{Model}} & \multicolumn{6}{c|}{\textbf{Training period 1950-2000}} & \multicolumn{6}{c|}{\textbf{Training period 1960-2000}} \tabularnewline
			\cline{3-14}
			& & \multicolumn{3}{c|}{\textbf{Male}} & \multicolumn{3}{c|}{\textbf{Female}} & \multicolumn{3}{c|}{\textbf{Male}} & \multicolumn{3}{c|}{\textbf{Female}} \tabularnewline
			\cline{3-14}
			& &  $RMSE_{(m)}$ & $PICP_{(m)}$ & $MPIW_{(m)}$ & $RMSE_{(m)}$ & $PICP_{(m)}$ & $MPIW_{(m)}$ & $RMSE_{(m)}$ & $PICP_{(m)}$ & $MPIW_{(m)}$ & $RMSE_{(m)}$ & $PICP_{(m)}$ & $MPIW_{(m)}$ \tabularnewline
			\hline	
			\multirow{2}{*}{\textbf{\textit{Australia}}} & \multirow{1}{*}{LC} & 0.227 & $\bm{1}$  & 	$\bm{0.534}$ &  $\bm{0.091}$ & 	0.944 & 0.267  & 0.175	& $\bm{1}$ & 	$\bm{0.478}$ & $\bm{0.084}$	 &	0.944 & 0.265 \tabularnewline 
			& \multirow{1}{*}{LC-LSTM} & $\bm{0.110}$ & 0.944   &  0.295 & 0.142  & 	0.944 	& $\bm{0.407}$ &  $\bm{0.116}$ & 0.944  &  0.280 &  0.097 & $\bm{1}$  & $\bm{0.394}$ \tabularnewline 
			\hline		
			\multirow{2}{*}{\textbf{\textit{Japan}}} & \multirow{1}{*}{LC} &  0.071	&  0.667 &	$\bm{0.180}$  & 0.255 & 0 & 	0.173 & $\bm{0.063}$ & 0.722 & 0.150 &	0.155  & 0.056  & 	0.105 \tabularnewline 
			& \multirow{1}{*}{LC-LSTM} & $\bm{0.062}$ &  $\bm{0.722}$ &  	0.143  & $\bm{0.077}$ &  $\bm{0.444}$	& $\bm{0.254}$  & 0.073 &  \textbf{0.944} & 	\textbf{0.243}  &  $\bm{0.061}$  & $\bm{0.667}$  &  $\bm{0.115}$ \tabularnewline 
			\hline
			\multirow{2}{*}{\textbf{\textit{Spain}}} & \multirow{1}{*}{LC} & 	0.200 &  0.333 & 0.153 & $\bm{0.104}$   & 	0.611 & 0.179 & 0.228 & $\bm{0.333}$ &	0.136 & $\bm{0.067}$ 	 & 0.722	&	0.174 \tabularnewline 
			& \multirow{1}{*}{LC-LSTM} & $\bm{0.161}$ & $\bm{0.556}$  &  $\bm{0.276}$  &  0.502  &  $\bm{0.944}$	& $\bm{0.489}$  & 	$\bm{0.205}$  & 0.278  &  $\bm{0.215}$  & 0.073  & $\bm{0.944}$  &  $\bm{0.259}$ \tabularnewline 
			\hline 
			\multicolumn{1}{c}{} & \multicolumn{1}{c}{} & \multicolumn{11}{c}{} & \multicolumn{1}{c}{ } \tabularnewline
			\multicolumn{1}{c}{} & \multicolumn{1}{c}{} & \multicolumn{11}{c}{} & \multicolumn{1}{c}{ } \tabularnewline
			
			
			\multicolumn{1}{c}{} & \multicolumn{1}{c}{} & \multicolumn{11}{c}{\LARGE$\bm{x=65}$} & \multicolumn{1}{c}{ } \tabularnewline
			\cline{1-14}
			\multirow{3}{*}{\textbf{Country}} & \multirow{3}{*}{\textbf{Model}} & \multicolumn{6}{c|}{\textbf{Training period 1950-2000}} & \multicolumn{6}{c|}{\textbf{Training period 1960-2000}} \tabularnewline
			\cline{3-14}
			& & \multicolumn{3}{c|}{\textbf{Male}} & \multicolumn{3}{c|}{\textbf{Female}} & \multicolumn{3}{c|}{\textbf{Male}} & \multicolumn{3}{c|}{\textbf{Female}} \tabularnewline
			\cline{3-14}
			& &  $RMSE_{(m)}$ & $PICP_{(m)}$ & $MPIW_{(m)}$ & $RMSE_{(m)}$ & $PICP_{(m)}$ & $MPIW_{(m)}$ & $RMSE_{(m)}$ & $PICP_{(m)}$ & $MPIW_{(m)}$ & $RMSE_{(m)}$ & $PICP_{(m)}$ & $MPIW_{(m)}$ \tabularnewline
			\hline	
			\multirow{2}{*}{\textbf{\textit{Australia}}} & \multirow{1}{*}{LC} & 0.157 &  	1 & $\bm{0.672}$  &  0.061 & 0.944 & 0.283 & 0.106 & 1 & \textbf{0.623} & 0.058 &	1 & 	0.293 \tabularnewline 
			& \multirow{1}{*}{LC-LSTM} & $\bm{0.056}$ & 1 & 	0.371	&  	0.061  &  $\bm{1}$	& $\bm{0.431}$ & 	$\bm{0.043}$  & 1  & 0.365  & \textbf{0.052}  & 1  & \textbf{0.436}  \tabularnewline 
			\hline		
			\multirow{2}{*}{\textbf{\textit{Japan}}} & \multirow{1}{*}{LC} & 0.054 &  $\bm{1}$ & $\bm{0.177}$ &  0.160 & 0.444 & 0.178 & 0.063  & 0.833 & 0.161 & 	0.151  & 0.333 & 	0.128 \tabularnewline 
			& \multirow{1}{*}{LC-LSTM} & $\bm{0.035}$ & 0.944  &  0.141 & \textbf{0.077}  &  \textbf{1}	&  	\textbf{0.262} & \textbf{0.029}  & \textbf{1}	 & \textbf{0.261}  & 	\textbf{0.028}  & \textbf{1}  &  \textbf{0.141} \tabularnewline 
			\hline
			\multirow{2}{*}{\textbf{\textit{Spain}}} & \multirow{1}{*}{LC} & 0.097 & 0.278  & 0.157 &  	0.079 & 	0.778  & 0.206 & 0.106 & 0.222 &  0.158 & 0.073 & 0.889	& 0.229 \tabularnewline 
			& \multirow{1}{*}{LC-LSTM} & $\bm{0.060}$ &  $\bm{1}$  & $\bm{0.285}$  &  \textbf{0.66}  &  \textbf{1}	& \textbf{ 0.568} &  \textbf{0.080} &  \textbf{0.889} & \textbf{0.249}  &  \textbf{0.068} &  \textbf{0.944} & \textbf{0.340}  \tabularnewline 
			\hline
			\multicolumn{1}{c}{} & \multicolumn{1}{c}{} & \multicolumn{11}{c}{} & \multicolumn{1}{c}{ } \tabularnewline
			\multicolumn{1}{c}{} & \multicolumn{1}{c}{} & \multicolumn{11}{c}{} & \multicolumn{1}{c}{ } \tabularnewline
			
			\multicolumn{1}{c}{} & \multicolumn{1}{c}{} & \multicolumn{11}{c}{\LARGE$\bm{x=85}$} & \multicolumn{1}{c}{ } \tabularnewline
			\cline{1-14}
			\multirow{3}{*}{\textbf{Country}} & \multirow{3}{*}{\textbf{Model}} & \multicolumn{6}{c|}{\textbf{Training period 1950-2000}} & \multicolumn{6}{c|}{\textbf{Training period 1960-2000}} \tabularnewline
			\cline{3-14}
			& & \multicolumn{3}{c|}{\textbf{Male}} & \multicolumn{3}{c|}{\textbf{Female}} & \multicolumn{3}{c|}{\textbf{Male}} & \multicolumn{3}{c|}{\textbf{Female}} \tabularnewline
			\cline{3-14}
			& &  $RMSE_{(m)}$ & $PICP_{(m)}$ & $MPIW_{(m)}$ & $RMSE_{(m)}$ & $PICP_{(m)}$ & $MPIW_{(m)}$ & $RMSE_{(m)}$ & $PICP_{(m)}$ & $MPIW_{(m)}$ & $RMSE_{(m)}$ & $PICP_{(m)}$ & $MPIW_{(m)}$ \tabularnewline
			\hline	
			\multirow{2}{*}{\textbf{\textit{Australia}}} & \multirow{1}{*}{LC} & \textbf{0.053} & 0.944   & \textbf{0.344} &  \textbf{0.032} & 1 & 	0.191  & \textbf{0.039} & 0.944 & \textbf{0.319} & 0.033 &	1 & 	0.194 \tabularnewline 
			& \multirow{1}{*}{LC-LSTM} & 0.056 & 0.944   &  0.190  &  0.033  &  1	&  \textbf{0.292} & 0.049 &  0.944 & 0.187  & \textbf{0.026} & 1  & \textbf{0.289}  \tabularnewline 
			\hline		
			\multirow{2}{*}{\textbf{\textit{Japan}}} & \multirow{1}{*}{LC} & 	\textbf{0.030} &  	\textbf{0.889} & \textbf{0.134} & 	\textbf{0.050}  & \textbf{0.778} &  	0.142 & 	0.040 & 0.944 & 0.133 & \textbf{0.071} & 0.444	& 	0.115 \tabularnewline 
			& \multirow{1}{*}{LC-LSTM} & 0.034 &  0.778 &  0.107 & 	0.171   & 0.500 	&  \textbf{0.209} & \textbf{0.029} & 0.944  &  \textbf{0.215} & 	0.080  & 0.444  &  \textbf{0.126} \tabularnewline 
			\hline
			\multirow{2}{*}{\textbf{\textit{Spain}}} & \multirow{1}{*}{LC} & 0.082 &  0.333 & 0.113 &  \textbf{0.059} & 0.611 &	0.122 & 0.086 & 0.278 &  0.116  & 0.057 & 0.833	& 0.150 \tabularnewline 
			& \multirow{1}{*}{LC-LSTM} & \textbf{0.052} &  \textbf{1} & \textbf{0.204} & 0.447   &  \textbf{1}	&  \textbf{0.335}  & \textbf{0.066}  &  \textbf{0.944}	 & \textbf{0.183}  &  \textbf{0.048} & \textbf{1}  & \textbf{ 0.223} \tabularnewline 
			\hline
		\end{tabular}
	\end{adjustbox}
\end{table}

Finally, we appreciate the LC model performances in uncertainty estimation for the Australian males. It is worth notes the LC greater probability coverage, as well as its interval width. Nevertheless, the latter hints some questions about the LC predictions suitability on the long-run. For instance, \autoref{fig:pi_males_long} displays a 50-years prediction for the Australian males aged 65, for both training periods. 

\begin{figure}[H]
	\centering
	\caption{Australian Males. PI ($\alpha=5\%$) for $x=65$. Training period: 1950-2000 (left), 1960-2000 (right). Forecasting period: 2001-2050. \label{fig:pi_males_long}}
	\includegraphics[width=0.45\linewidth,height=0.25\textheight]{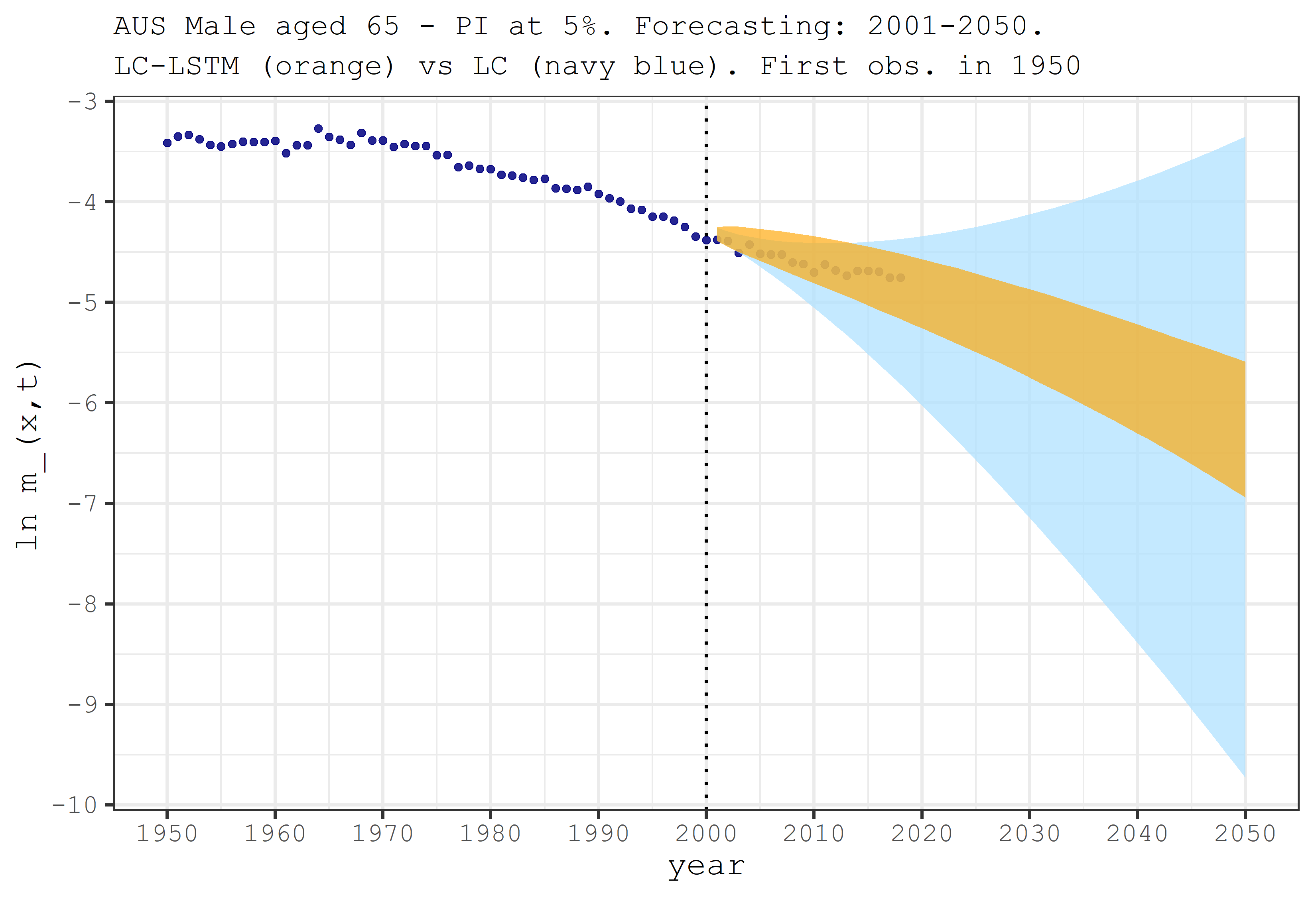}  \quad
	\includegraphics[width=0.45\linewidth,height=0.25\textheight]{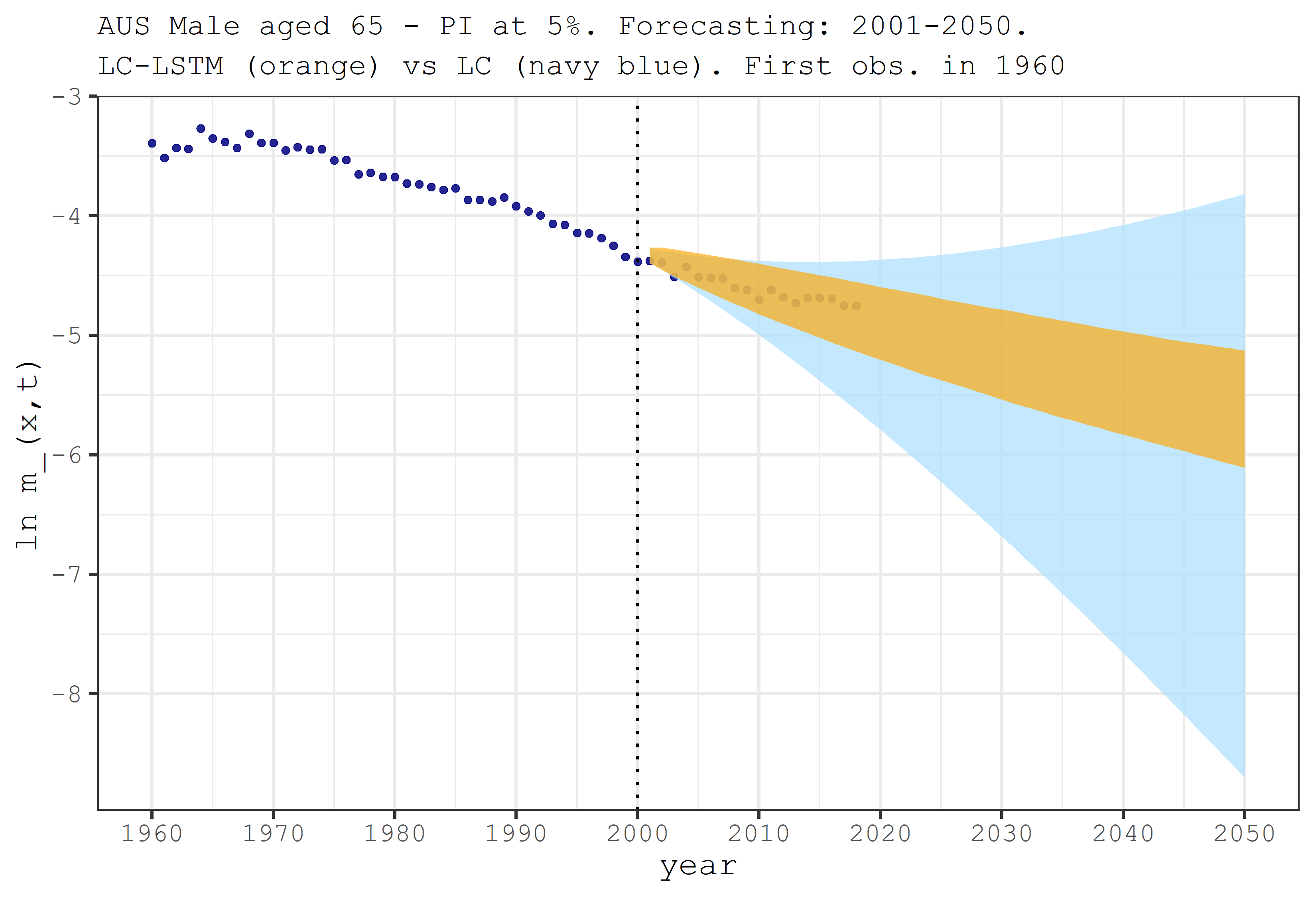} \\
\end{figure}

Given the observed mortality up to the forecasting year, the LC model seems to propose uncertainty levels not consistent with the historical mortality dynamic. Looking at the training period 1960-2000, we observe an overall death rates reduction about 61\%. In the following 40 years of projection, the LC model establishes a further reduction in death rates around 96\%, in the case of the PI lower bound, or a possible increase of 68\%, considering the PI upper bound. Considering the training period 1950-2000, these evidences are strengthened. Referring to the LC-LSTM model, the mortality estimates assume greater consistency with historical observations. In particular, the LC-LSTM produces a 40-year decrease in mortality between 82\%, considering the PI lower bound, and 46\% according to the PI upper bound.

Moreover, inspecting \autoref{fig:pi_males_long} we stress how the learning period length impacts on the long-run network forecasts. As aforementioned, the two learning periods considered show different accelerations in mortality decline. Fitting the LSTM model on the interval 1960-2000, the network learns the fundamental linear decrease of mortality such that a coherent PI shape is predicted over the forecasting horizon. As opposite, the interval 1950-2000 points up a non-linear behaviour due to the longevity accelerations in the period 1950-1960. In this case, the LSTM is able in extrapolating a coherent mortality range with the historical observation, allowing for biological plausibility but believing in a more marked increase in longevity. In light of this, we do not question the robustness of the model, rather we emphasize its ability to extrapolate the fundamental pattern from the observed data. The selection of the historical sample to fit the mortality model depends on the aware modeller expert judgment, given  the population under investigation. As necessary support, the modeller must be consider adequately criteria to approve its period choices. As suggested by \cite{Cairns2011}, is crucial evaluate qualitative ex ante criteria, such as biological reasonableness, plausibility of predicted levels of uncertainty and model robustness. At same time, ex post quantitative criteria, such as performance metrics in Section \ref{sec:metriche}, are indispensable to address forecasts in a backtesting exercise (see for instance \cite{dowd}). 
Following both qualitative and quantitative criteria, overall our analyses demonstrate how both models are biologically regular in projecting mortality. The discerning factor between the two models is the plausibility about foreseen uncertainty levels, especially for long-term forecasts. Hence, the proposed model integration favours the predictive improvements in the LC framework, as well as proven in most cases by performance indicators. Finally, we suggest the interval 1960-2000 as the most suitable training period for the LSTM calibration. In fact, it is plausible to believe that the reduction in mortality will continue to occur in a fair linear way over time and at different ages, properly reflecting the demographic trend observed since the 1960s.              

\bigskip

\section{Conclusions}
\label{sec:conclusioni}
Nowadays, the mortality forecast is still a major challenge for actuaries and demographers. Obviously, different populations might show diverse mortality scenarios, and a performing mortality model for a population might be not adequate on another one. Not surprisingly, the collection of mortality models in literature is far-reaching. Recently, new methodological advances in mortality forecasting has been proposed, grounding on machine and deep learning techniques, mainly based on Neural Network models. The latter have established a forecasting approach with high predictive accuracy in point estimates. However, to the best of our knowledge, machine and deep learning literature in mortality forecasting lack for studies about uncertainty estimation. In light of this, the present work formalizes a deep learning integration of the LC model, both in terms of point prediction and prediction interval.   
%
Our proposal allows, at the same time, to represent the mortality surface through a canonical age-period model and to predict the future mortality realizations extrapolating the temporal mortality dynamic from data. The resulting LC-LSTM model poses a compromise between interpretation of the mortality phenomenon and high precision in anticipating its future realizations. Moreover, exploiting both the NN ensemble paradigm and noise analysis, we are able to produce a mortality density forecast. 
From our empirical investigation, we highlight the LC-LSTM capacity to produce forecasts both biologically consistent and plausible in uncertainty levels w.r.t. the historical observations, also in the long-run. The latter features is crucial in actuarial assessments, as in the evaluation of annuities products or to appraise the pension system sustainability.
Therefore, our proposal establish a reliable improvement of the LC model in term of predictive prowess, posing an innovative approach within mortality literature. The proposed framework might represent a prominent practice in the field of longevity forecasting, also for actuarial business tasks.

\newpage

\appendix

\section*{Appendix A: Forecasting mortality with the LC model}
\label{sec:appendix_A}
We recall the fundamental forecasting equations of the LC model. According to the predictor in eq.(\ref{eq:predictor_LC}), over the forecast horizon $\mathcal{T'}$ the terms $k_{t}$ are usually modeled by a random walk. Generalizing, we consider an ARIMA(p,d,q) process, then the realizations of $k_{t}$ over $\mathcal{T'}$ originate from the following equation:
\begin{equation}
\label{eq:kt_arima}
\bigtriangledown^dk_{t_n+h}=h\delta+\sum_{i=1}^{p}\phi_i\bigtriangledown^dk_{(t_n+h)-i}+\sum_{j=1}^{q}\theta_j\epsilon_{(t_n+h)-j}+\sum_{k=1}^{h}\epsilon_{t_n+k}, \quad h=1,\ldots,s \\
\end{equation}
where the sum of the errors are normally distributed, that is $\sum_{k=1}^{h}\epsilon_{t_n+k} \sim \mathcal{N}\left(0,h^2\sigma^2_{\epsilon}\right)$. In this case, the LC model provides the following point predictions, for all $t \in \mathcal{T'} $:
\begin{equation}
\label{eq:point_kt_arima}
\log \hat{m}_{x,t}=\mathbb{E}\left(\log m_{x,t}\right) = \hat{\alpha}_x+\hat{\beta}_x \left(h\delta+\sum_{i=1}^{p}\phi_i\bigtriangledown^dk_{(t_n+h)-i}+\sum_{j=1}^{q}\theta_j\epsilon_{(t_n+h)-j}\right).\\
\end{equation}
Since errors are gaussian, the prediction interval at a fixed confidence level $\alpha \in (0,1)$ is:
\begin{equation}	
\label{eq:interval_kt_arima}
\log \hat{m}_{x,t} \pm \hat{\beta}_x\sqrt{h}\sigma_{\epsilon} z_{\frac{\alpha}{2}}. \\
\end{equation}
where $z_{\alpha}$ is the $\alpha$-quantile of a Standard Normal distribution.

\newpage

\section*{Appendix B: Statistical tests to check the noise randomness and normality}
\label{test_norm}

\begin{table}[H]
	\centering
	\renewcommand*{\arraystretch}{1.1}
	\caption{Statistical tests for noise in the training set. Males.}
	\begin{adjustbox}{max width=\textwidth}
	\begin{tabular}{|c|c|cc|cc|}
		\hline
		\multirow{2}{*}{\textbf{Country}} & \multirow{2}{*}{\textbf{Test}} & \multicolumn{2}{c|}{\textbf{Training period 1950-2000}} & \multicolumn{2}{c|}{\textbf{Training period 1960-2000}} \tabularnewline
		\cline{3-6}
		& &  Statistics value & p-value & Statistics value & p-value \tabularnewline
		\hline	
		\multirow{4}{*}{\textbf{\textit{Australia}}} & \multirow{1}{*}{Shapiro-Wilk} &0.96352 & 0.12489$^{\star \star \star }$&  0.98379 & 0.82539$^{\star\star\star}$ \tabularnewline 
		& \multirow{1}{*}{D'Agostino-Pearson} &1.62692 & 0.44332$^{\star\star\star}$& 0.85534 & 0.65203$^{\star\star\star}$ \tabularnewline 
		& \multirow{1}{*}{Jarque-Bera} & 1.55177 & 0.46030$^{\star\star\star}$ & 0.64381 & 0.72477$^{\star\star\star}$ \tabularnewline 
		& \multirow{1}{*}{ADF} & -3.05447 & 0.15132$^{\star\star\star}$ & -2.58739 & 0.34294$^{\star\star\star}$ \tabularnewline 
		\hline		
		\multirow{4}{*}{\textbf{\textit{Japan}}} & \multirow{1}{*}{Shapiro-Wilk} & 0.96193 & 0.10710$^{\star\star\star}$  & 0.97511 & 0.51356$^{\star\star\star}$  \tabularnewline 
		& \multirow{1}{*}{D'Agostino-Pearson} & 8.05556 & 0.01781$^{\star}$ & 1.45996 & 0.48192$^{\star\star\star}$ \tabularnewline 
		& \multirow{1}{*}{Jarque-Bera} & 7.35771 & 0.02525$^{\star}$  &  1.20406 & 0.54770$^{\star\star\star}$ \tabularnewline 
		& \multirow{1}{*}{ADF} & -3.49574 & 0.05128$^{\star\star}$  & -2.73088 & 0.28662$^{\star\star\star}$ \tabularnewline 
		\hline	
		\multirow{4}{*}{\textbf{\textit{Spain}}} & \multirow{1}{*}{Shapiro-Wilk} & 0.97654 & 0.41696$^{\star\star\star}$ & 0.95790 & 0.14191$^{\star\star\star}$  \tabularnewline 
		& \multirow{1}{*}{D'Agostino-Pearson} & 1.83229 & 0.40006$^{\star\star\star}$ & 2.82652 & 0.24335$^{\star\star\star}$ \tabularnewline 
		& \multirow{1}{*}{Jarque-Bera} & 1.05350 & 0.59052$^{\star\star\star}$ & 2.31446 & 0.31436$^{\star\star\star}$ \tabularnewline 
		& \multirow{1}{*}{ADF} & -7.55942 & 0.01000  & -4.11879 & 0.01516$^{\star}$ \tabularnewline 
		\hline
	\end{tabular}
	\end{adjustbox}
\small{P-value significance level: $> 0.01 ^{ \star}$, $> 0.05 ^{\star \star}$, $ > 0.1 ^{\star \star \star} $.}
\end{table}
\begin{table}[H]
	\centering
	\renewcommand*{\arraystretch}{1.1}
	\caption{Statistical tests for noise in the training set. Females.}
	\begin{adjustbox}{max width=\textwidth}
	\begin{tabular}{|c|c|cc|cc|}
		\hline
        \multirow{2}{*}{\textbf{Country}} & \multirow{2}{*}{\textbf{Test}} & \multicolumn{2}{c|}{\textbf{Training period 1950-2000}} & \multicolumn{2}{c|}{\textbf{Training period 1960-2000}} \tabularnewline
		\cline{3-6}
		& &  Statistics value & p-value & Statistics value & p-value \tabularnewline
		\hline
		\multirow{4}{*}{\textbf{\textit{Australia}}} & \multirow{1}{*}{Shapiro-Wilk} & 0.96907 & 0.21209$^{\star\star\star}$  & 0.96724 & 0.29319$^{\star\star\star}$ \tabularnewline 
		& \multirow{1}{*}{D'Agostino-Pearson} & 2.52531 & 0.28290$^{\star\star\star}$ & 0.78319 & 0.67598$^{\star\star\star}$ \tabularnewline 
		& \multirow{1}{*}{Jarque-Bera} &1.78204  & 0.41024$^{\star\star\star}$ & 0.60740 & 0.73808$^{\star\star\star}$ \tabularnewline 
		& \multirow{1}{*}{ADF} & -3.07190 & 0.14432$^{\star\star\star}$ & -2.50033 & 0.37711$^{\star\star\star}$ \tabularnewline 
		\hline
		\multirow{4}{*}{\textbf{\textit{Japan}}} & \multirow{1}{*}{Shapiro-Wilk} & 0.97452 & 0.34985$^{\star\star\star}$   & 0.98888 & 0.95815$^{\star\star\star}$   \tabularnewline 
		& \multirow{1}{*}{D'Agostino-Pearson} & 3.12195 & 0.20993 $^{\star\star\star}$ & 0.79814 & 0.67094$^{\star\star\star}$  \tabularnewline 
		& \multirow{1}{*}{Jarque-Bera} &2.09605 & 0.35063$^{\star\star\star}$&  0.62112 & 0.73303$^{\star\star\star}$ \tabularnewline 
		& \multirow{1}{*}{ADF} & -5.14239 & 0.01000  & -3.89596 & 0.02383$^{\star}$ \tabularnewline 
		\hline
		\multirow{4}{*}{\textbf{\textit{Spain}}} & \multirow{1}{*}{Shapiro-Wilk} & 0.93640 & 0.02619$^{\star}$  & 0.97970 & 0.67844$^{\star\star\star}$ \tabularnewline 
		& \multirow{1}{*}{D'Agostino-Pearson} & 8.69754 & 0.01292$^{\star}$  & 1.74855 & 0.41716$^{\star\star\star}$  \tabularnewline 
		& \multirow{1}{*}{Jarque-Bera} & 7.56206 & 0.02280$^{\star}$   & 1.20753 & 0.54675$^{\star\star\star}$ \tabularnewline 
		& \multirow{1}{*}{ADF} &  -5.80177 & 0.01000 & -3.46488 & 0.06172$^{\star\star\star}$ \tabularnewline 
		\hline
		\end{tabular}
	\end{adjustbox}
\small{P-value significance level: $> 0.01 ^{ \star}$, $> 0.05 ^{\star \star}$, $ > 0.1 ^{\star \star \star} $.}
\end{table}

\newpage


\end{document}